\documentclass[journal]{IEEEtran}   %% IEEE期刊双栏版本
\usepackage{multirow}%added by yxw 20161022
\usepackage{diagbox}%added by yxw 20161022
\usepackage{amssymb}
\usepackage{amsmath}
\usepackage{graphicx}
\usepackage{cite}
\usepackage{citesort}
\usepackage{subfigure}
\usepackage{graphicx,epstopdf}
\usepackage{epsfig}	
\usepackage{cite,graphicx,amsmath,amssymb}
\usepackage{comment}

\usepackage{amssymb}
\usepackage{amsmath}
\usepackage{cite}
\usepackage{url}
\usepackage{xcolor}
\usepackage{cite,graphicx,amsmath,amssymb}
\usepackage{subfigure}
\usepackage{citesort}
\usepackage{fancyhdr}
\usepackage{mdwmath}
\usepackage{mdwtab}
\usepackage{caption}
\usepackage{amsthm}
\usepackage{lipsum}

\newtheorem{theorem}{Theorem}

\newtheorem{lemma}{Lemma}

\newtheorem{corollary}{Corollary}

\newtheorem{remark}{Remark}  % added by yxw 20160929
\newtheorem{proposition}{Proposition}

\makeatletter
\def\ScaleIfNeeded{%
\ifdim\Gin@nat@width>\linewidth \linewidth \else \Gin@nat@width
\fi } \makeatother

\begin{document}

\title{\Huge{Spatially Random Relay Selection for Full/Half-Duplex Cooperative NOMA Networks}}

\author{ Xinwei~Yue,~\IEEEmembership{Student Member,~IEEE,} Yuanwei\ Liu,~\IEEEmembership{Member,~IEEE,}
 Shaoli~Kang, Arumugam~Nallanathan,~\IEEEmembership{Fellow,~IEEE},
 and Zhiguo~Ding,~\IEEEmembership{Senior Member,~IEEE}

\thanks{This work was supported by National High Technology Research and Development Program of China (863 Program, 2015AA01A709).
This work was supported in part by the U.K. Engineering and Physical Sciences Research Council (EPSRC) under Grant EP/M016145/2 and The work of Z. Ding was supported by the UK EPSRC under grant numbers EP/L025272/1 and EP/N005597/1.
}
\thanks{X. Yue is with School of Electronic and Information Engineering, Beihang university, Beijing 100191,
China (email: xinwei$\_$yue@buaa.edu.cn).}
\thanks{Y. Liu and A. Nallanathan are with School of Electronic Engineering and Computer Science, Queen Mary University of London, London E1 4NS, U.K. (email: \{yuanwei.liu, a.nallanathan\}@qmul.ac.uk).}
\thanks{S. Kang is with State Key Laboratory of Wireless Mobile Communications, China Academy of Telecommunications
Technology(CATT), Beijing 100094, China and also with School of Electronic and Information Engineering,
Beihang University, Beijing 100191, China (email: kangshaoli@catt.cn).}
\thanks{Z. Ding is with the Department of Electrical Engineering, Princeton University, Princeton, USA and also with
the School of Computing and Communications, Lancaster University, Lancaster LA1 4YW, U.K. (e-mail: z.ding@lancaster.ac.uk).
Part of this work has been accepted in IEEE GLOBECOM 2017 \cite{Yue2016selection}.}
 } 
%Part of this work has been submitted to IEEE ICC 2017 \cite{Yue2016Non}.
%Z. Ding is with School of Computing and Communications, Lancaster University,
 %UK (e-mail: z.ding@lancaster.ac.uk).
%\author{
%\IEEEauthorblockN{  Xinwei~Yue\IEEEauthorrefmark{1}, Yuanwei Liu\IEEEauthorrefmark{3}, Shaoli Kang\IEEEauthorrefmark{2},Arumugam Nallanathan\IEEEauthorrefmark{3}   }
% Yuanwei Liu
% Arumugam Nallanathan
%\IEEEauthorblockA{\\
%\IEEEauthorrefmark{1}School of Electronic and Information Engineering
%, Beihang University, Beijing, China\\
%\IEEEauthorrefmark{2} China Academy of Telecommunication Technology, Beijing, China\\
%\IEEEauthorrefmark{3} School of Information and Communication Engineering, Beijing University of Posts and Telecommunications, Beijing, %China\\
%\IEEEauthorrefmark{3} King's College London, London, UK
% } }

\maketitle

%\textcolor[rgb]{0.00,0.00,1.00}{}
\begin{abstract}
This paper investigates the impact of relay selection (RS) on the performance of cooperative non-orthogonal multiple access (NOMA),
where relays are capable of working in either full-duplex (FD) or half-duplex (HD) mode.
A number of relays (i.e., $K$ relays) are uniformly distributed within the disc.
A pair of RS schemes are considered insightfully: 1) Single-stage RS (SRS) scheme; and 2) Two-stage RS (TRS) scheme. In order to characterize the performance of these two RS schemes, new closed-form expressions for both exact and asymptotic outage probabilities are derived. Based on analytical results, the diversity orders achieved by the pair of RS schemes for FD/HD cooperative NOMA are obtained. Our analytical results reveal that: i) The FD-based RS schemes obtain a zero diversity order, which is due to the influence of loop interference (LI) at the relay; and ii) The HD-based RS schemes are capable of achieving a diversity order of $K$, which is equal to the number of relays. Finally, simulation results demonstrate that: 1) The FD-based RS schemes have better outage performance than HD-based RS schemes in the low signal-to-noise radio (SNR) region;
2) As the number of relays increases, the pair of RS schemes considered are capable of achieving the lower outage probability; and 3) The outage behaviors of FD/HD-based NOMA SRS/TRS schemes are superior to that of random RS (RRS) and orthogonal multiple access (OMA) based RS schemes.
%\textcolor[rgb]{0.00,0.00,1.00}{}

% FD/HD-based NOMA SRS/TRS schemes have better outage behaviors than random RS (RRS) and orthogonal multiple access (OMA) based RS schemes.

\end{abstract}
\begin{keywords}
{D}ecode-and-forward, full/half-duplex, relay selection, non-orthogonal multiple access %, stochastic geometry
\end{keywords}
\section{Introduction}
%\textcolor[rgb]{0.98,0.00,0.00}{(We need a spotlight starting sentence here)}
With the rapid advancement in the wireless communication technology, the fifth generation (5G) mobile communication networks have attracted a great deal of attention~\cite{Li67306795G,METIS,Liu2018Beyond}. In particular, three major families of new radio (NR) usage scenarios, i.e., massive machine type communications (mMTC), enhanced mobile broadband (eMBB) and ultra-reliable and low-latency communications (URLLC) are proposed to satisfy the different requirements for 5G networks. To improve system throughput and achieve enhanced spectrum efficiency of 5G networks, non-orthogonal multiple access (NOMA) has been considered
to be a promising candidate technique and identified for 3GPP Long Term Evolution (LTE) \cite{Ding2017Mag}.
The core idea of NOMA is able to multiplex additional users in the same physical resource. More specifically, the superposition coding scheme is employed at the transmitting end, where the linear superposition of signals of multiple users is formed to be the transmit signal. The successive interference cancellation (SIC) procedure is carried out by the receiving end who has the better channel conditions \cite{Cai2017Modulation}. Furthermore, downlink multiuser superposition transmission scheme (MUST) \cite{MUST1} which is the special case of NOMA has found application in wireless standard.

Hence numerous excellent Contributions have surveyed the performance of point-to-point NOMA in wireless networks in
\cite{Ding6868214,Pairing7273963,Shi7438933,Xu7506136}.
To evaluate the performance of downlink NOMA, the closed-form expressions of outage probability and ergodic rate for NOMA were derived in \cite{Ding6868214} by use of the bounded
path loss model. Furthermore, the authors of \cite{Pairing7273963} have studied the impact of user pairing on the performance of NOMA, where both the outage performance of fixed power allocation based NOMA (F-NOMA) and cognitive radio based NOMA (CR-NOMA) schemes were characterized. By considering user grouping and decoding order selection, the outage balancing among users was investigated \cite{Shi7438933}, in which the closed-form expressions of optimal decoding order and power allocation for downlink NOMA were derived. In \cite{Xu7506136}, the authors researched the outage behavior of downlink NOMA for the case where each NOMA user only feed back one bit of its channel state information (CSI) to a base station (BS). It was shown that NOMA is capable of providing higher fairness for multiple users compared to conventional opportunistic one-bit feedback. As a further advance, there is a paucity of research treaties
on investigating the application of point-to-point NOMA systems. In \cite{Liu2016TVT}, the authors analyzed the outage behavior of large-scale underlay CR for NOMA with the aid of stochastic geometry. To emphasize physical layer security (PLS),
the authors in \cite{Physical7812773} discussed the PLS issues of NOMA, where the secrecy outage probabilities were derived for both single-antenna and multiple-antenna scenarios, respectively. Recently, the NOMA-based wireless cashing strategies were introduced in \cite{Ding2017NOMA}, in which two cashing phases, i.e., content pushing and content delivery, are characterized in terms of caching hit probability.
Additionally, explicit insights for understanding the performance of uplink NOMA have been provided in \cite{Zhang7390209,Tabassum2016}. In \cite{Zhang7390209}, the novel uplink power control protocol was proposed for
the single-cell uplink NOMA. In large-scale cellular networks, the performance of multi-cell uplink NOMA was characterized
in terms of coverage probability using the theory of Poisson cluster process \cite{Tabassum2016}.

Cooperative communication is a promising approach to overcome signal fading arising from multipath propagation
as well as obtain the higher diversity \cite{laneman2004cooperative}. Obviously, combining cooperative communication technique and NOMA is the research topic which has sparked of wide interest in\cite{Ding2014Cooperative,Kim2015Capacity,Liu7445146SWIPT,Wan7842026}. The concept of cooperative NOMA was initially proposed for downlink transmission in \cite{Ding2014Cooperative}, where the nearby user with better channel conditions was viewed as decode-and-forward (DF) relay to deliver the information for the distant users. Driven by these, authors in \cite{Kim2015Capacity} analyzed the achievable data rate of NOMA systems for DF relay over Rayleigh fading channels. On the standpoint of tackling spectrum efficiency and energy efficiency, in \cite{Liu7445146SWIPT}, the application of simultaneous wireless information and power transfer (SWIPT) to NOMA with randomly deployed users was investigated using stochastic geometry. In \cite{Wan7842026}, NOMA based dual-hop relay systems were addressed, where both statistical CSI and instantaneous CSI were considered for the networks. On the other hand, the outage performance of NOMA for a variable gain amplify-and-forward (AF) relay was characterized over Nakagami-$m$ fading channels in \cite{Men7454773}. With the emphasis on imperfect CSI, authors studied the system outage behavior of AF relay for NOMA networks in \cite{Men7752764}. Additionally, the authors of \cite{Yuexinwei7812773} analyzed the outage performance of a fixed gain based AF relay for NOMA systems over Nakagami-$m$ fading channels.
%stochastic geometry based models were surveyed

Above existing contributions on cooperative NOMA are all based on the assumption of half-duplex (HD) relay, where the
communication process was completed in two slots \cite{laneman2004cooperative}. To further improve the bandwidth usage efficiency of system, full-duplex (FD) relay technology is a promising solution which can simultaneously receive and transmit the signal in the same frequency band \cite{Riihonen5089955}. Nevertheless, FD operation suffers from residual loop self-interference (LI), which is usually modeled as a fading channel \cite{Duarte6656015}. Particularly, FD relay technologies in \cite{Zhang2015Full} have been discussed from the view of self-interference cancellation, protocol design and relay selection for 5G networks. To maximize the weighted sum throughput of system, the design of resource allocation algorithm for FD multicarrier NOMA (MC-NOMA) was investigated in \cite{Sun7812683}, where a FD BS was capable of serving downlink and uplink users in the meantime.
The recent findings in FD operation considered for cooperative NOMA were surveyed in \cite{Ding2016FD,Zhong7572025}.
The performance of FD device-to-device (D2D) based cooperative NOMA was characterized in terms of outage probability in \cite{Ding2016FD}. Considering the influence of imperfect self-interference, the authors in \cite{Zhong7572025} investigated the performance of FD-based DF relay for NOMA, where the expressions of outage probability and achievable sum rate for two NOMA users were derived.

%% 介绍relay selection 研究成果 attractive
Applying relay selection (RS) technique to cooperative communication systems is a straightforward and effective approach
for taking advantages of space diversity and improving spectral efficiency. The following research contributions have
surveyed the RS schemes for two kinds of operation modes: HD and FD. For HD mode, the authors of \cite{Jing4801494} derived the diversity of single RS scheme and investigated the complexity of multiple RS scheme by exhaustive search. It was shown that these RS schemes are capable of providing full diversity order.
%the authors in \cite{Jing4801494} have investigated the several RS schemes which can achieve full diversity.
Furthermore, in \cite{Zlatanov7084188}, the ergodic rate was studied with a buffer-aided relay scheme for HD-based
single RS network. Additionally, the application of RS scheme to cognitive DF relay networks was discussed in \cite{Liu6940288}.
For FD mode, assuming the availability of different instantaneous CSI, the authors analyzed the RS problem of AF cooperative system in \cite{Krikidis2012Full}. It was worth noting that FD-based RS scheme converges to an error floor and
obtains a zero diversity order. The performance of DF RS scheme was characterized in terms of outage probability for the CR networks in \cite{Zhong6949656}. Very recently, two-stage RS scheme was proposed for HD-based cooperative NOMA in \cite{Ding7482785}, where the RS scheme considered was capable of realizing the maximal diversity order.
%Additionally, the outage performance of NOMA with partial RS scheme was evaluated in \cite{Lee7794655}.

\subsection{Motivations and Contributions}
While the aforementioned significant contributions have laid a solid foundation for the understanding of cooperative NOMA and RS techniques, the RS technique for cooperative NOMA networks is far from being well understood.
It is worth pointing out that from a practical perspective,
the requirements of Internet of Things (IoT) scenarios, i.e, link density, coverage enhancement and small packet service
are capable of being supported through the RS schemes. One of the best relays is selected from $K$ relays as the BS's helper to forward the information.
In \cite{Ding7482785}, the two-stage RS scheme is capable of achieving the minimal outage probability and obtaining the maximal diversity order, but only HD-based RS for cooperative NOMA was considered. To the best of our knowledge, there are no existing works to investigate the RS scheme for FD cooperative NOMA networks. Moreover, the spatial impact of RS on the performance of FD cooperative NOMA was not examined in \cite{Ding7482785}.
Motivated by these, we specifically consider a pair of RS schemes for FD/HD NOMA networks, namely single-stage RS (SRS) scheme and two-stage RS (TRS) scheme, where the locations of relays are modeled by invoking the uniform distribution. More specifically, in the SRS scheme, the data rate of distant user is ensured to select a relay as its helper to forward the information. In the TRS scheme, on the condition of ensuring the data rate of distant user, we serve the nearby user with data rate as large as possible for selecting a relay. Based on the proposed schemes, the primary contributions can
be summarized as follows:
\begin{enumerate}
  \item
   We investigate the outage behaviors of two RS schemes (i.e., SRS scheme and TRS scheme) for FD NOMA networks. We derive the closed-form and asymptotic expressions of outage probability for FD-based NOMA RS schemes.
   %Based on the analytical results, we also derive the corresponding asymptotic outage probabilities and obtain diversity orders.
    Due to the influence of residual LI at relays, a pair of FD-based NOMA RS schemes converge to an error floor in the high signal-to-noise radio (SNR) region and provide zero diversity order.
   \item
    We also derive the closed-form expressions of outage probability for two HD-based NOMA RS schemes. To get more insights,
    the asymptotic outage probabilities of HD-based NOMA RS schemes are derived. We observe that with the number of relays increasing, the lower outage probability can be achieved for HD-based NOMA RS schemes. We confirm that the HD-based NOMA RS schemes are capable of providing the diversity order of $K$, which is equal to the number of relays.
   \item
    We show that the outage behaviors of FD-based NOMA SRS/TRS schemes are superior to that of HD-based NOMA SRS/TRS schemes in the low SNR region rather than in the high SNR region. Furthermore, we confirm that the FD/HD-based NOMA TRS/SRS schemes are capable of providing better outage performance compare to random RS (RRS) and orthogonal multiple access (OMA) based RS schemes. Additionally, we analyze the system throughput in delay-limited transmission mode based on the outage probabilities derived.
\end{enumerate}

%We confirm that the outage performance of HD-based RS schemes considered outperforms random HD-based RS schemes.

\subsection{Organization and Notation}
The rest of the paper is organized as follows. In Section \ref{System Model}, the network model of the RS schemes for FD/HD
NOMA is set up. New analytical and approximate expressions of outage probability for the RS schemes are derived in Section \ref{Section_III}. In Section \ref{Numerical Results}, numerical results are presented for performance evaluation and comparison. Section \ref{Conclusions} concludes the paper.

The main notations of this paper is shown as follows: $\mathbb{E}\{\cdot\}$ denotes expectation operation; ${f_X}\left(  \cdot  \right)$ and ${F_X}\left(  \cdot  \right)$ denote the probability density function (PDF) and the cumulative distribution function (CDF) of a random variable $X$.
% $ \propto $ represents ``be proportional to".

\section{Network Model}\label{System Model}
In this section, the network and signal models are presented.
Additionally, the criterions of a pair of RS schemes in the networks considered are introduced for FD/HD NOMA.
\subsection {Network Description}
Consider a downlink cooperative NOMA scenario consisting of one BS, $K$ relays (${R_i}$ with $1 \le i \le K$) and a pair of users (i.e., the nearby user ${D_1}$ and distant user $D_{2}$), as shown in Fig. \ref{System Model1}. To reduce the complexity of NOMA system, multiple users can be divided into several groups and the NOMA protocol is carried out in each group \cite{Pairing7273963,Liu7982794}. The groups between each other are orthogonal.
We assume that the BS is located at the origin of a disc, denoted by $ {\cal D}$ and the radius of disc is ${R_{\cal D}}$.
In addition, $K$ relays  are uniformly distributed within $ {\cal D}$ \cite{Ding6868214}.
%\textcolor[rgb]{0.00,0.00,1.00}{The locations of relays are modeled as homogeneous binomial point processes (HBPPs) \cite{Ding6868214,Stochastic2012} and the number of relays (i.e., $K$ relays) in $ {\cal D}$ are assumed to be Binomial distribution.}
%%
%\textcolor[rgb]{0.00,0.00,1.00}{The locations of relays are modeled as homogeneous poisson point processes (HPPPs) with density ${\lambda _\psi }$ \cite{Stochastic2012}. Hence the number of relays (i.e., $K$ relays) in $ {\cal D}$ are assumed to be Poisson distribution ${\Pr }\left( {K = i} \right) = \frac{{{\mu ^i}{e^{ - \mu }}}}{{i!}}$, where $\mu$ is the mean measure, i.e, $\mu  = \pi R_{\cal D}^2{\lambda _\psi }$.}
The DF protocol is employed at each relay and only one relay is selected to assist BS conveying the information to the NOMA users in each time slot. To enable FD operation, each relay is equipped with one transmit antenna and one receive antenna, while the BS and users have a single antenna\footnote{Note that multiple antennas equipped by the BS and relays will further suppress the self-interference and enhance the performance of the NOMA-based RS schemes. Additionally, more sophisticated assumption of antennas at relay, i.e., omni-directional and directional antennas \cite{Riihonen5985554,Kolodziej7426862} can be developed for further evaluating the performance of the networks considered. However, these are beyond the scope of this treatise.}, respectively.
All wireless channels\footnote{It is assumed that perfect CSI can be obtained, our future work will relax this idealized assumption. Furthermore, we note that relaxing the setting of Rayleigh fading channels (i.e., Nakagami-$m$ fading channels considered in \cite{Yuexinwei7812773}) will provide a more general system setup, which are set aside for our future work.} in the scenario considered are assumed to be independent non-selective block Rayleigh fading and are disturbed by additive white Gaussian noise with mean power $N_{0}$.
${h_{SR_{i}}} \sim {\cal C}{\cal N}\left( {0,1} \right)$, ${h_{R_{i}D_{1}}}\sim {\cal C}{\cal N}\left( {0,1} \right)$, and
${h_{R_{i}D_{2}}} \sim {\cal C}{\cal N}\left( {0,1} \right)$ denote the complex channel coefficient of $BS \rightarrow R_{i}$,
$R_{i} \rightarrow D_{1}$, and $R_{i} \rightarrow D_{2}$ links, respectively. $d_1$ and $d_2$ denote the distance from the BS
to $D_1$ and $D_2$, respectively. Assuming that an imperfect self-interference cancellation scheme is employed
at each relay such as \cite{Krikidis2012Full} and the corresponding LI is modeled as a Rayleigh fading
channel with coefficient ${h_{{LI}}} \sim {\cal C}{\cal N}\left( {0,{\Omega _{{LI}}}} \right)$.
As stated in \cite{Ding7482785}, two NOMA users are classified into the nearby user and distant user by their quality of service (QoS) not sorted by their channel conditions.
More particularly, via the assistance of the best relay selected, the QoS requirements of NOMA users can be supported effectively for the IoT scenarios (i.e., small packet business and telemedicine service) \cite{Ding2016MIMO}. Hence we assume that $D_1$ can be served opportunistically and $D_2$ needs to be served quickly for small packet with a lower target data rate. As a further example, $D_1$ is to download a movie or carry out some background tasks and so on; $D_2$ can be a medical health sensor which is to send the pivotal safety information containing in a few bytes, such as blood pressure, pulse and heart rates.
%\textcolor[rgb]{0.00,0.00,1.00}{More particularly, in the IoT scenarios, the users have various requirements, i.e., link density, coverage enhancement and low power consumption, small packet business and telemedicine service. With the assistance of the best relay selected, the requirements of NOMA users can be supported effectively. Hence we assume that $D_1$ can be served opportunistically and $D_2$ needs to be served quickly for small packet with a lower target data rate \cite{Ding2016MIMO}. As a further example, $D_1$ is to download a movie or carry out some background tasks and so on; $D_2$ can be a medical health sensor which is to send the pivotal safety information containing in a few bytes, such as blood pressure, pulse and heart rates.}

%link density, coverage enhancement and low power consumption,
\begin{figure}[t!]
    \begin{center}
        \includegraphics[width=3.5in,  height=1.8in]{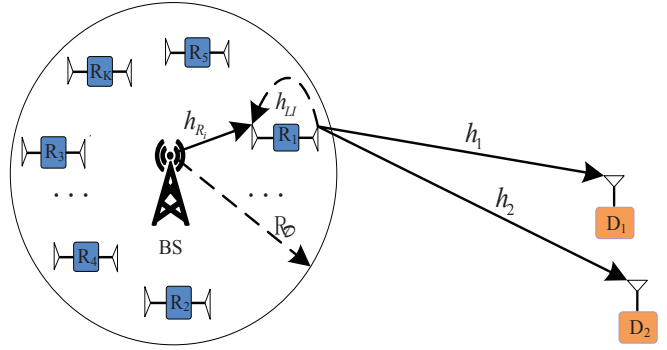}
        \caption{An illustration of RS scheme for downlink FD/HD cooperative NOMA networks.}
        \label{System Model1}
    \end{center}
\end{figure}
\subsection{Signal Model}
During the $l$-th time slot, $l=1,2,3,...$, the BS sends the superposed signal $\sqrt {{a_1}{P_s}} {x_1}\left[ l \right] + \sqrt {{a_2}{P_s}} {x_2}\left[ l \right]$ to the relay on the basis of NOMA principle \cite{Ding6868214}, where ${{x_1}}$ and $x_{2}$ are the normalized signal for $D_{1}$ and $D_{2}$, respectively, i.e, $\mathbb{E}\{x_{1}^2\}= \mathbb{E}\{x_{2}^2\}=1$. ${{a_1}}$ and $a_{2}$ are the corresponding power allocation coefficients. Practically speaking, to stipulate better fairness and QoS requirements between the users \cite{Ding2016MIMO}, we assume that ${a_2} \ge {a_1}$ with $a_{1} +a_{2} = 1$. The LI signal exists at the relay due to it works in FD mode. Therefore the observation at the $i$th relay $R_{i}$ is given by
\begin{align}\label{the signal expression for relay}
{y_{{R_i}}} =& {h_{R_i}}(\sqrt {{a_1}{P_s}} {x_1}\left[ l \right] + \sqrt {{a_2}{P_s}} {x_2}\left[ l \right])  \nonumber \\
&+ {h_{LI}}\sqrt {\varpi {P_r}} {x_{LI}}\left[ {l - {l_d} } \right]+ {n_{{R_i}}},
\end{align}
where ${h_{R_i}} = \frac{{{h_{S{R_i}}}}}{{\sqrt {1 + d_{S{R_i}}^\alpha } }}$, ${d_{SR_{i}}}$ is the distance
between the BS and $R_{i}$ and $\alpha $ denotes the path loss exponent. $\varpi$ is the switching operation factor, where
$\varpi=1$ and $\varpi=0$ denote the relay working in FD mode and HD mode, respectively. According to the practical usage scenarios, we can select the different duplex mode. It is worth noting that in FD mode, it is capable of improving the spectrum efficiency, but will suffer from the LI signals. On the contrary, in HD mode, this situation can be avoided precisely.
${{P_s}}$ and ${{P_r}}$ denote normalized transmission power (i.e., ${{P_s}}={{P_r}}=1$) at the BS and $R_{i}$, respectively. ${x_{LI}}[l - {l_d} ]$ denotes the LI signal with ${\mathop{\mathbb{E}}\nolimits} [|{x_{LI}}{|^2}] = 1$ and an integer ${l_d} $ denotes processing delay at $R_{i}$ with ${l_d}  \ge 1$. ${n_{{R_i}}}$ denotes the Gaussian noise at $R_{i}$.

%Due to the assumption ${a_1} \le {a_2}$, SIC\footnote{In this paper, we assume that perfect SIC is employed, our future work will relax this assumption.} is employed at $R_{i}$ to first decode the signal $x_{2}$ of $D_{2}$, since $R_{i}$ has a less interference-infested signal to decode the signal $x_{1}$ of $D_{1}$.
Based on NOMA protocol, SIC\footnote{In this paper, we assume that perfect SIC is employed, our future work will relax this assumption.} is employed at  $R_{i}$ to first decode the signal $x_{2}$ of $D_{2}$ having a higher power allocation factor, since $R_{i}$ has a less interference-infested signal to decode the signal $x_{1}$ of $D_{1}$.
Based on this, the received signal-to-interference-plus-noise ratio (SINR) at $R_{i}$ to detect
$x_{2}$ and $x_{1}$ are given by
\begin{align}\label{the SINR1 for relay}
{\gamma _{{D_2} \to {R_i}}} = \frac{{\rho {{\left| {{h_{R_i}}} \right|}^2}{a_2}}}{{\rho {{\left|
 {{h_{R_i}}} \right|}^2}{a_1} + \rho \varpi {{\left| {{h_{LI}}} \right|}^2} + 1}},
\end{align}
and
\begin{align}\label{the SINR2 for relay}
{\gamma _{{D_1} \to {R_i}}} = \frac{{\rho {{\left| {{h_{R_i}}} \right|}^2}{a_1}}}{{\rho \varpi {{\left| {{h_{LI}}} \right|}^2} + 1}},
\end{align}
respectively, where $\rho  = \frac{{{P_s}}}{{{N_0}}}$ is the transmit SNR.

Assuming that $R_i$ is capable of decoding the two NOMA user's information, i.e,
satisfying the following conditions, 1) $\log \left( {1 + {\gamma _{{D_1} \to {R_i}}}} \right) \ge {R_{{D_1}}}$; and 2) $\log \left( {1 + {\gamma _{{D_2} \to {R_i}}}} \right) \ge {R_{{D_2}}}$,
where ${R_{{D_1}}}$ and ${R_{{D_2}}}$ are the target rate for $D_{1}$ and $D_{2}$, respectively.
Therefore the observation at $D_{j}$ can be expressed as
\begin{align}\label{the signal expression for users}
{y_{{D_j}}} = {h_j}(\sqrt {{a_1}{P_r}} {x_1}\left[ {l - {l_d} } \right]
 + \sqrt {{a_2}{P_r}} {x_2}\left[ {l - {l_d} } \right]) + {n_{{D_j}}},
\end{align}
where ${h_j} = \frac{{{h_{{R_i}{D_j}}}}}{{\sqrt {1 + d_{{R_i}{D_j}}^\alpha } }}$, $d_{R_{i}D_{j}}$ is the distance between
$R_{i}$ and $D_{j}$ (assuming $d_{R_iD_j} \gg d_{SR_i}$); ${d_{{R_{i}D_j}}} = \sqrt {d_{S{R_i}}^2 + d_j^2 - 2{d_{S{R_i}}}{d_j}\cos \left( {{\theta _i}} \right)}$, $j \in \left( {1,2} \right)$. ${\theta _i}$ denotes the angle $\angle {D_j}S{R_i}$; ${n_{{D_j}}}$ denotes the Gaussian noise at $D_{j}$.

In similar, assuming that SIC can be also invoked successfully by $D_{1}$ to detect the signal of $D_2$ having a higher transmit power, who has less interference. Hence the received SINR at $D_{1}$ to detect $x_{2}$ can be given by
\begin{align}\label{the SINR1 for D1 to detect D2}
{\gamma _{{D_2} \to {D_1}}} = \frac{{\rho {{\left| {{h_1}} \right|}^2}{a_2}}}{{\rho {{\left| {{h_1}} \right|}^2}{a_1} + 1}}.
\end{align}
Then the received SINR at $D_{1}$ to detect its own information is given by
\begin{align}\label{the SINR2 for D1 to detect D1}
{\gamma _{{D_1}}} = \rho {\left| {{h_1}} \right|^2}{a_1}.
\end{align}
The received SINR at $D_{2}$ to detect $x_{2}$ is given by
\begin{align}\label{the SINR3 for D2}
{\gamma _{{D_2}}} = \frac{{\rho {{\left| {{h_2}} \right|}^2}{a_2}}}{{\rho {{\left| {{h_2}} \right|}^2}{a_1} + 1}}.
\end{align}
Note that the fixed power allocation coefficients for two NOMA users are considered in the networks. Reasonable power control and optimizing the mode of power allocation can further enhance the performance of the RS schemes, which may be investigated in our future work.

\subsection{Relay Selection Schemes}
In this subsection, we consider a pair of RS schemes for FD/HD NOMA, which are detailed in the following.
\subsubsection{Single-stage Relay Selection}
Prior to the transmissions, a relay can be randomly selected by the BS as its helper to forward the information. The aim of SRS scheme is to maximize the minimum data rate of $D_2$ for FD/HD NOMA.
More specifically, the size of data rate for $D_2$ depends on three kinds of data rates, such as 1) the data rate for the relay $R_i$ to detect $x_2$; 2) The data rate for $D_1$ to detect $x_2$; and 3) the data rate for $D_2$ to detect its own signal $x_2$.
Among the relays in the network, based on \eqref{the SINR1 for relay}, \eqref{the SINR1 for D1 to detect D2} and \eqref{the SINR3 for D2}, the SRS scheme activates a relay, i.e.,
\begin{align}\label{Selection two for maxmix far user data rate}
 {i_{SRS}^*} =& \mathop {\arg }\limits_i \max \left\{ {\min \left\{ {\log \left( {1 + {\gamma _{{D_2} \to {R_i}}}} \right),\log \left( {1 + {\gamma _{{D_2} \to {D_1}}}} \right)} \right.}\right., \nonumber\\
& \begin{array}{*{20}{c}}
   {} & {} & {} & {} & {} & {}  \\
\end{array}\begin{array}{*{20}{c}}
   {\left. {\left. {\log \left( {1 + {\gamma _{{D_2}}}} \right)} \right\}},i \in S_R^{1} \right\}}, & {}  \\
\end{array}
\end{align}
where ${S_R^{1}}$ denotes the number of relays in the network. Note that FD/HD-based SRS schemes inherit advantage to ensure the data rate of $D_2$, where the application of small packets can be achieved.
\subsubsection{Two-stage Relay Selection}
The TRS scheme mainly include two stages for FD/HD NOMA: 1) In the first stage, the
target data rate of $D_{2}$ is to be satisfied; and 2) In the second stage, on the condition that the data rate
of $D_{2}$ is ensured, we serve $D_{1}$ with data rate as large as possible. Hence the first stage activates the relays
that satisfy the following condition
\begin{align}\label{Two-stage relay selection strategy}
 S_R^{2} =& \left\{ {\log \left( {1 + {\gamma _{{D_2} \to {R_i}}}} \right) \ge {R_{{D_2}}},} \right.\log \left( {1 + {\gamma _{{D_2} \to {D_1}}}} \right) \ge {R_{{D_2}}}, \nonumber \\
 \begin{array}{*{20}{c}}
   {} & {}  \\
\end{array}\begin{array}{*{20}{c}}
   {}  \\
\end{array}&\begin{array}{*{20}{c}}
   {\left. {\log \left( {1 + {\gamma _{{D_2}}}} \right) \ge {R_{{D_2}}},1 \le i \le K} \right\}},\\
\end{array}
\end{align}
where the size of $S_R^{2}$ is defined as $\left| {S_R^{2}} \right|$.  %$S_R^{2}$ denotes the subset of the relays and

Among the relays in $S_R^{2}$, the second stage selects a relay to convey the information which can maximize the data
rate of $D_{1}$ and is expressed as
\begin{align}
 i_{TRS}^* =& \mathop {\arg }\limits_i \max \left\{ {\min \left\{ {\log \left( {1 + {\gamma _{{D_1} \to {R_i}}}} \right),} \right.} \right. \nonumber \\
 &\left. {\begin{array}{*{20}{c}}
   {} & {} & {} & {} & {} & {}  \\
\end{array}\begin{array}{*{20}{c}}
   {\left. {\log \left( {1 + {\gamma _{{D_1}}}} \right)} \right\},i \in S_R^{2}}  \\
\end{array}} \right\}.
\end{align}
As can be observed from the above explanations, the merit of FD/HD-based TRS schemes is that in addition to guarantee the data rate of $D_2$, the BS can support $D_1$ to carry out some background tasks, i.e., downloading a movie or multimedia files.
\section{Performance evaluation}\label{Section_III}
In this section, the performance of this pair of RS schemes are characterized in terms of outage probability as well as the delay-limited throughput for FD/HD NOMA networks.
%we present the details of two kinds of RS schemes in the networks considered. More particularly,
\subsection{Single-stage Relay Selection Scheme}\label{The First Relay Selection Scheme}
According to NOMA protocol, the complementary events of outage for SRS scheme can be explained as: 1)
The relay $i_{SRS}^*$ can detect the signal $x_{2}$ of $D_{2}$; and 2) while the signal $x_{2}$ can be successfully detected
at $D_{1}$ and $D_{2}$, respectively. From the above descriptions, the outage probability of SRS scheme for FD NOMA can be expressed as follows:
\begin{align}\label{theorem expression for relay selection two}
P_{SRS}^{FD} = \prod\limits_{i = 1}^K {\left( {1 - \Pr \left( {{W_i} > \gamma _{t{h_2}}^{FD}} \right)} \right)},
\end{align}
where ${W_i} = \min \left\{ {{\gamma _{{D_2} \to {R_i}}},{\gamma _{{D_2} \to {D_1}}},{\gamma _{{D_2}}}} \right\}$ and $\varpi=1$.
$\gamma _{t{h_2}}^{FD}=2^{R_{D_{2}}}-1$ with $R_{D_{2}}$ being the target rate of $D_{2}$.

The following theorem provides the outage probability of SRS scheme for FD NOMA.
\begin{theorem}\label{theorem relay selection one}
The closed-form expression of outage probability for FD-based NOMA SRS scheme can be approximated as follows:
%\textcolor[rgb]{0.00,0.00,1.00}{Conditioned on the $HBPPs$, the outage probability of the FD-based NOMA SRS scheme can be approximated as follows:}
\begin{align}\label{OP derived for FD relay selection two}
 P_{SRS}^{FD} \approx &  \left[ {1 - \left( {1 - \frac{\pi }{{2N}}\sum\limits_{n = 1}^N {\sqrt {1 - \phi _n^2} }
  \left( {{\phi _n} + 1} \right)} \right.} \right. \nonumber \\
& {\left. {\left. { \times \left( {1 - \frac{{{e^{ - {c_n}\tau }}}}{{{1 + \varpi \rho  \tau {c_n}{\Omega _{{\rm{LI}}}}}}}} \right)}
 \right){e^{ - \left( {1 + d_1^\alpha }
  \right)\tau  - \left( {1 + d_2^\alpha } \right)\tau }}} \right]^K} ,
\end{align}%{\tau^{'}}
where $\varpi=1$, ${\tau}{\rm{ = }}\frac{{\gamma _{t{h_2}}^{FD}}}{{\rho \left( {{a_2} - {a_1}\gamma _{t{h_2}}^{FD}} \right)}}$ with
${a_2} > {a_1}\gamma _{t{h_2}}^{FD}$. ${c_n} = 1 + {\left( {\frac{{{R_{\cal D}}}}{2}\left( {{\phi _n} + 1} \right)} \right)^\alpha }$, ${\phi _n} = \cos \left( {\frac{{2n - 1}}{{2N}}\pi } \right)$ and $N$ is a parameter to ensure a complexity-accuracy tradeoff.
\end{theorem}
\begin{proof}
See Appendix~A.
\end{proof}

\begin{corollary}\label{corollary: OP derived for the second HD relay selection}
Upon substituting $\varpi=0$ into \eqref{OP derived for FD relay selection two}, the approximate expression of outage probability for HD-based NOMA SRS scheme is given by
\begin{align}\label{OP expression for the second HD relay selection}
 P_{SRS}^{HD} \approx & \left[ {1 - \left( {1 - \frac{\pi }{{2N}}\sum\limits_{n = 1}^N {\sqrt {1 - \phi _n^2} } \left( {{\phi _n} + 1} \right)} \right.} \right. \nonumber \\
& {\left. {\left. { \times \left( {1 - {e^{ - {\tau _1}{c_n}}}} \right)} \right){e^{ - \left( {1 + d_1^\alpha } \right)
 {\tau _1} - \left( {1 + d_2^\alpha } \right){\tau _1}}}} \right]^K},
\end{align}
where ${\tau _1}{\rm{ = }}\frac{{\gamma _{t{h_2}}^{HD}}}{{\rho \left( {{a_2} - {a_1}\gamma _{t{h_2}}^{HD}} \right)}}$ with ${a_2} > {a_1}\gamma _{t{h_2}}^{HD}$ and $\gamma _{t{h_2}}^{HD}=2^{2{R_{D_{2}}}}-1$ with ${R_{D_{2}}}$ being the target rate of $D_{2}$.
\end{corollary}

\subsection{Two-stage Relay Selection Scheme}
In the case of TRS scheme, the overall outage event can be expressed \cite{Ding7482785} as follows:
\begin{align}
\varphi  = {\varphi _1}  \cup  {\varphi _2},
\end{align}
where ${\varphi _1}$ denotes the outage event that relay $ i_{TRS}^*$ cannot detect $x_{2}$, or neither $D_{1}$ and $D_{2}$ can
detect the $x_{2}$ correctively, and ${\varphi _2}$ denotes the outage event that either of $ i_{TRS}^*$ and $D_{1}$ cannot
detect $x_{1}$ while three nodes can detect $x_{2}$ successfully.

As a consequence, the outage probability of TRS scheme for FD NOMA can be expressed as follows:
\begin{align}\label{the expression of OP for TRS scheme for FD}
P_{TRS}^{FD} = \Pr \left( {{\varphi _1}} \right) + \Pr \left( {{\varphi _2}} \right).
\end{align}
On the basis of analytical results in \eqref{The First Relay Selection Scheme}, the first outage probability in \eqref{the expression of OP for TRS scheme for FD}
is approximated as
\begin{align}\label{The first OP derived for FD TS}
 \Pr \left( {{\varphi _1}} \right) \approx & \left[ {1 - \left( {1 - \frac{\pi }{{2N}}\sum\limits_{n = 1}^N
  {\sqrt {1 - \phi _n^2} } \left( {{\phi _n} + 1} \right)} \right.} \right. \nonumber \\
 &{\left. {\left. { \times \left( {1 - \frac{{{e^{ - {c_n}{\tau}}}}}{{{1 + \rho \varpi \tau {c_n}{\Omega _{{\rm{LI}}}}}}}} \right)}
  \right){e^{ - \left( {1 + d_1^\alpha } \right){\tau} - \left( {1 + d_2^\alpha } \right){\tau }}}} \right]^K},
\end{align}
where $\varpi=1$.

In order to calculate the second outage probability, ${\mathop{\rm P}\nolimits} \left( {{\varphi _2}} \right)$ can be further expressed as
\begin{align}\label{the second term}
&\Pr \left( {{\varphi _2}} \right) = \Pr \left( {{\Lambda _1},\left| {S_R^{2}} \right| > 0} \right)
+\Pr \left( {{\Lambda _2},{{\bar \Lambda }_1},\left| {S_R^{2}} \right| > 0} \right) ,
\end{align}
where ${{\Lambda _1}}$ denotes the outage event that the relay $i_{TRS}^*$ cannot detect $x_{1}$ and
${{{\bar \Lambda }_1}}$ denotes the corresponding complementary event of ${{\Lambda _1}}$.
${{\Lambda _2}}$ denotes that $D_{1}$ cannot detect $x_{1}$. The first term in the above equation is given by
\begin{align}\label{term1}
 & \Pr \left( {{\Lambda _1},\left| {S_R^{2}} \right| > 0} \right)  \\
 &=\Pr \left( {\log \left( {1 + {\gamma _{{D_1} \to {R_{i_{TRS}^*}}}}} \right) < {R_{{D_1}}},\left| {S_R^{2}} \right| > 0} \right)\nonumber.
\end{align}
The second term in \eqref{the second term} is given by
\begin{align}\label{term2}
&\Pr \left( {{\Lambda _2},{{\bar \Lambda }_1},\left| {S_R^{2}} \right| > 0} \right) = { \Pr } \left( {\log \left( {1 + {\gamma _{{D_1}}}} \right) < {R_{{D_1}}},} \right. \nonumber\\
& \begin{array}{*{20}{c}}
   {} & {} & {} & {} & {} & {} & {\log \left( {1 + {\gamma _{{D_1} \to {R_{i_{TRS}^*}}}}} \right) > {R_{{D_1}}}\left. {,\left|
    {S_R^{2}} \right| > 0} \right)}.
\end{array}
\end{align}

Combining \eqref{term1} with \eqref{term2},  the second outage probability in \eqref{the expression of OP for TRS scheme for FD} can be expressed as
\begin{align}\label{the seconde term OP for TS}
& \Pr \left( {{\varphi _2}} \right) = \Pr \left( {\log \left( {1 + {\gamma _{{D_1} \to
{R_{i_{TRS}^*}}}}} \right) < {R_{{D_1}}},\left| {S_R^{2}} \right| > 0} \right) \nonumber \\
 &\begin{array}{*{20}{c}}
   {} & {} & {} & { + \Pr \left( {\log \left( {1 + {\gamma _{{D_1}}}} \right) < {R_{{D_1}}},} \right.} \nonumber  \\
\end{array} \\
& \begin{array}{*{20}{c}}
   {} & {} & {} & {\log \left( {1 + {\gamma _{{D_1} \to {R_{i_{TRS}^*}}}}} \right) > {R_{{D_1}}}\left. {,\left| {S_R^{2}}
   \right| > 0} \right)},
\end{array}
\end{align}
where $\varpi=1$.

To derive the closed-form expression of outage probability for TRS scheme in \eqref{the seconde term OP for TS}, we define
\begin{align}
{s_i} = \min \left\{ {\log \left( {1 + {\gamma _{{D_1} \to {R_i}}}} \right),\log \left( {1 + {\gamma _{{D_1}}}} \right)} \right\},
\end{align}
and
\begin{align}
{s_{{ i_{TRS}^*}}} = \max \left\{ {{s_k},\forall k \in S_R^{2}} \right\},
\end{align}
respectively.
The probability $\Pr \left( {{\varphi _2}} \right)$ can be given by
\begin{align}
 \Pr \left( {{\varphi _2}} \right) =&\Pr \left( {\min \left\{ {\log \left( {1 + {\gamma _{{D_1} \to {R_{ i_{TRS}^*}}}}} \right),} \right.} \right. \nonumber\\
 &\left. {\begin{array}{*{20}{c}}
   {} & {} & {} & {\left. {\log \left( {1 + {\gamma _{{D_1}}}} \right)} \right\}} \nonumber \\
\end{array} < {R_{{D_1}}},\left| {S_R^{2}} \right| > 0} \right) \nonumber\\
  =& \Pr \left( {{s_{{ i_{TRS}^*}}} < {R_{{D_1}}},\left| {S_R^{2}} \right| > 0} \right).
\end{align}
The above probability can be further expressed as
\begin{align}\label{TS OP expression}
 \Pr \left( {{\varphi _2}} \right) =& \sum\limits_{k = 1}^K {\Pr
 \left( { {s_{{ i_{TRS}^*}}} < {R_{{D_1}}},\left| {S_R^{2}} \right| = k} \right)}  \nonumber \\
 =& \sum\limits_{k = 1}^K {\Pr \left( {{s_{i_{TRS}^*}} < {R_{{D_1}}}|\left| {S_R^2} \right| = k} \right)\Pr \left( {\left| {S_R^2} \right| = k} \right)} \nonumber \\
 = & \sum\limits_{k = 1}^K {{{\left[ {\underbrace {F\left( {{R_{{D_1}}}} \right)}_{{\Theta _1}}} \right]}^k}\underbrace {\Pr \left( {\left| {S_R^2} \right| = k} \right)}_{{\Theta _2}}} .
\end{align}

For selecting a relay at random from ${S_R^2}$, denoted by relay ${i_R}$, let us now turn our attention to the derivation of ${s_{{i_R}}}$'s CDF (i.e., ${F\left( {{R_{{D_1}}}} \right)}$) in the following lemma.
%the corresponding CDF of ${s_{{i}}}$ (i.e., ${F\left( {{R_{{D_1}}}} \right)}$) is provided in the following lemma in detail.
%From the above equation, we observe that the two probabilities must be further calculated.
Define these two probabilities at the right hand side of \eqref{TS OP expression} by ${{\rm{\Theta }}_1}$ and ${{\rm{\Theta }}_2}$, respectively.
\begin{lemma}\label{lemma:1}
The conditional probability in \eqref{TS OP expression} can be approximated as follows:
\begin{align}\label{the approx express Theta1}
{{\rm{\Theta }}_1} \approx
 \frac{{{{ M}_1} + {{ M}_2} + {{ M}_3}}}{{{e^{ - \left( {1 + d_1^\alpha } \right){\tau }}}\left( {1 - \Delta \left( {1 - \chi {e^{ - {c_n}{\tau}}}} \right)} \right)}},
\end{align}
where $\Delta {\rm{ = }}\frac{\pi }{{2N}}\sum\limits_{n = 1}^N {\sqrt {1 - \phi _n^2} } \left( {{\phi _n} + 1}
    \right)$, $\theta {\rm{ = }}\max \left( {{\tau },\xi } \right)$, $\xi {\rm{ = }}\frac{{\gamma _{t{h_1}}^{FD}}}{{\rho {a_1}}}$,
     $\zeta  = \frac{{{c_n} + \left( {1 + d_1^\alpha } \right)}}{{\rho \varpi {c_n}}}$, $\chi  = \frac{1}{{1 + \rho \varpi {\tau } {c_n} {\Omega _{{\rm{LI}}}}}}$, $\psi  = \frac{1}{{1 + \rho \varpi \xi {c_n}  {\Omega _{{\rm{LI}}}}}}$, ${\rm T}  = \frac{{\left( {1 + d_1^\alpha } \right){e^{ - \left( {{c_n} + \left( {1 + d_1^\alpha } \right)} \right)\xi }}}}{{\rho \varpi {c_n} {\Omega _{{\rm{LI}}}}}}$, $\Phi  = \frac{{\left( {1 + d_1^\alpha } \right){e^{ - \left( {{c_n} + \left( {1 + d_1^\alpha } \right)} \right){\tau }}}}}{{\rho \varpi{c_n} {\Omega _{{\rm{LI}}}}}}$, ${{ M}_1} = {e^{ - \left( {1 + d_1^\alpha } \right)\theta }}\Delta \left( {\chi {e^{ - {c_n}{\tau }}} - \psi {e^{ - {c_n}\xi }}} \right)$, ${{M}_2} = \Delta \left( {\left( {{e^{ - \left( {1 + d_1^\alpha } \right){\tau }}} - {e^{ - \left( {1 + d_1^\alpha } \right)\xi }}} \right){e^{ - {c_n}{\tau }}}\chi  - {\rm T} {e^{\frac{\zeta }{{{\Omega _{{\rm{LI}}}}\psi }}}}{{\mathop{\rm E}\nolimits} _{\mathop{\rm i}\nolimits} }\left( {\frac{{ - \zeta }}{{{\Omega _{{\rm{LI}}}}\psi }}} \right)} \right. \\
 \left. { + \Phi {e^{\frac{{ - \zeta }}{{{\Omega _{{\rm{LI}}}}\chi }}}}{{\mathop{\rm E}\nolimits} _{\mathop{\rm i}\nolimits} }\left( {\frac{{ - \zeta }}{{{\Omega _{{\rm{LI}}}}\chi }}} \right)} \right) $,
 $ {M_3} = {e^{ - \left( {1 + d_1^\alpha } \right)\tau }} - {e^{ - \left( {1 + d_1^\alpha } \right)\xi }} - \Delta \left( {{e^{ - \left( {1 + d_1^\alpha } \right)\tau }} - {e^{ - \left( {1 + d_1^\alpha } \right)\xi }} - {\rm T}{e^{\frac{\zeta }{{{\Omega _{{\rm{LI}}}}\psi }}}}{{\mathop{\rm E}\nolimits} _{\mathop{\rm i}\nolimits} }\left( {\frac{{ - \zeta }}{{{\Omega _{{\rm{LI}}}}\psi }}} \right) + \Phi {e^{\frac{{ - \zeta }}{{{\Omega _{{\rm{LI}}}}\chi }}}}} \right. \\
 \left. { \times {{\mathop{\rm E}\nolimits} _{\mathop{\rm i}\nolimits} }\left( {\frac{{ - \zeta }}{{{\Omega _{{\rm{LI}}}}\chi }}} \right)} \right) $, $\gamma _{t{h_1}}^{FD}=2^{R_{D_{1}}}-1$ with $R_{D_{1}}$ being the target rate of $D_{1}$ and ${{\mathop{\rm E}\nolimits} _{\mathop{\rm i}\nolimits} }\left(  \cdot  \right)$ is the exponential integral function~\cite[Eq. (8.211.1)]{gradshteyn}.

\end{lemma}
\begin{proof}
See Appendix~B.
\end{proof}

On the other hand, there are $k$ relays in ${S_R^{2}}$ and the corresponding probability ${{\rm{\Theta }}_2}$ is given by
\begin{align}
 {{\rm{\Theta }}_2} =& \prod\limits_{m = 1}^{K - k} {K \choose
  k  } {\left( {1 - \Pr \left( {{\gamma _{{D_2} \to {R_i}}} > \gamma _{t{h_2}}^{FD}} \right)} \right.} \nonumber \\
 &\left. { \times \Pr \left( {{\gamma _{{D_1} \to {D_2}}} > \gamma _{t{h_2}}^{FD}} \right)\Pr \left( {{\gamma _{{D_2}}} > \gamma _{t{h_2}}^{FD}} \right)} \right) \nonumber\\
 &\times \prod\limits_{m = K - k + 1}^K {\Pr \left( {{\gamma _{{D_2} \to {R_i}}} > \gamma _{t{h_2}}^{FD}} \right)}  \nonumber\\
  &\times \Pr \left( {{\gamma _{{D_1} \to {D_2}}} > \gamma _{t{h_2}}^{FD}} \right)\Pr \left( {{\gamma _{{D_2}}} > \gamma _{t{h_2}}^{FD}} \right).
\end{align}

With the aid of \textbf{Theorem \ref{theorem relay selection one}}, the above probability can be further approximated as follows:
\begin{align*}\label{the approx express Theta2}
 {\Theta _2} \approx & {K \choose
  k  }\left[ {1 - \left( {1 - \frac{\pi }{{2N}}\sum\limits_{n = 1}^N {\sqrt {1 - \phi _n^2} } } \right.} \right.\left( {{\phi _n} + 1} \right) \nonumber\\
  & \times \left. {\left( {1 - \frac{{{e^{ - {c_n}{\tau }}}}}{{1 + \rho \varpi \tau {c_n}{\Omega _{{\rm{LI}}}}}}} \right)} \right){\left. {{e^{ - \left( {1 + d_1^\alpha } \right){\tau } - \left( {1 + d_2^\alpha } \right){\tau }}}} \right]^{K - k}}\nonumber \\
    \end{align*}
  \begin{align}
 & \times \left[ {\left( {1 - \frac{\pi }{{2N}}\sum\limits_{n = 1}^N {\sqrt {1 - \phi _n^2} } } \right.} \right.\left( {{\phi _n} + 1} \right) \nonumber\\
 & \times \left. {\left( {1 - \frac{{{e^{ - {c_n}{\tau }}}}}{{1 + \rho \varpi \tau {c_n}{\Omega _{{\rm{LI}}}}}}} \right)} \right){\left. {{e^{ - \left( {1 + d_1^\alpha } \right){\tau } - \left( {1 + d_2^\alpha } \right){\tau }}}} \right]^k}.
\end{align}

With the aid of \textbf{Lemma \ref{theorem relay selection one}}, combining \eqref{The first OP derived for FD TS},  \eqref{TS OP expression}, \eqref{the approx express Theta1} and
\eqref{the approx express Theta2} and applying some algebraic manipulations, the outage probability of TRS scheme for FD NOMA can be provided in the following theorem.
\begin{theorem}\label{theorem:2}
The closed-form expression of outage probability for the FD-based NOMA TRS scheme is approximated by \eqref{the last OP expression for TS} at the top of next page.
\begin{figure*}[!t]
\normalsize
\begin{align}\label{the last OP expression for TS}
P_{TRS}^{FD} \approx &\sum\limits_{k = 0}^K {
   K  \choose
   k
} \left[ {\frac{{{e^{ - \left( {1 + d_1^\alpha } \right)\theta }}\Delta \left( {\chi {e^{ - {c_n}\tau }} - \psi {e^{ - {c_n}\xi }}} \right)}}{{{e^{ - \left( {1 + d_1^\alpha } \right){\tau _1}}}\left( {1 - \Delta \left( {1 - \chi {e^{ - {c_n}\tau }}} \right)} \right)}} + \frac{{\Delta \left( {{e^{ - \left( {1 + d_1^\alpha } \right)\tau }} - {e^{ - \left( {1 + d_1^\alpha } \right)\xi }}} \right){e^{ - {c_n}\tau }}\chi }}{{{e^{ - \left( {1 + d_1^\alpha } \right)\tau }}\left( {1 - \Delta \left( {1 - \chi {e^{ - {c_n}\tau }}} \right)} \right)}}} \right. \nonumber\\
 & {\left. {\begin{array}{*{20}{c}}
   {} & { + \frac{{{e^{ - \left( {1 + d_1^\alpha } \right)\tau }} - {e^{ - \left( {1 + d_1^\alpha } \right)\xi }} - \Delta \left( {{e^{ - \left( {1 + d_1^\alpha } \right)\tau }} - {e^{ - \left( {1 + d_1^\alpha } \right)\xi }}} \right)}}{{{e^{ - \left( {1 + d_1^\alpha } \right)\tau }}\left( {1 - \Delta \left( {1 - \chi {e^{ - {c_n}\tau }}} \right)} \right)}}}  \\
\end{array}} \right]^k}{\left[ {1 - \left( {1 - \Delta \left( {1 - \chi {e^{ - {c_n}\tau }}} \right)} \right){e^{ - \left( {1 + d_1^\alpha } \right)\tau  - \left( {1 + d_2^\alpha } \right)\tau }}} \right]^{K - k}} \nonumber\\
 & \begin{array}{*{20}{c}}
   {} & { \times {{\left[ {\left( {1 - \Delta \left( {1 - \chi {e^{ - {c_n}\tau }}} \right)} \right){e^{ - \left( {1 + d_1^\alpha } \right)\tau  - \left( {1 + d_2^\alpha } \right)\tau }}} \right]}^k}}
\end{array}  .
\end{align}
\hrulefill \vspace*{0pt}
\end{figure*}
\end{theorem}

\begin{corollary}\label{corollary: OP derived for the second HD relay selection}
For the special case $\varpi=0$,  the approximate expression of outage probability for HD-based NOMA TRS scheme is given by \eqref{OP expression for the TS relay selection in HD mode} at the top of next page,
\begin{figure*}[!t]
\normalsize
\begin{align}\label{OP expression for the TS relay selection in HD mode}
 P_{TRS}^{HD} \approx & \sum\limits_{k = 0}^K {{K \choose
  k  }\left[ {\frac{{\Delta \frac{{{c_n}}}{{\left( {1 + d_1^\alpha } \right) + {c_n}}}\left( {{e^{ - \left( {\left( {1 + d_1^\alpha } \right) + {c_n}} \right){\tau _1}}} - {e^{ - \left( {\left( {1 + d_1^\alpha } \right) + {c_n}} \right){\xi _1}}}} \right)}}{{{e^{ - \left( {1 + d_1^\alpha } \right){\tau _1}}}\left( {1 - \Delta \left( {1 - {e^{ - {\tau _1}{c_n}}}} \right)} \right)}}} \right.} \nonumber \\
 &{\left. { + \frac{{{e^{ - \left( {1 + d_1^\alpha } \right){\tau _1}}} - {e^{ - \left( {1 + d_1^\alpha } \right){\xi _1}}} - \Delta \left[ {{e^{ - \left( {1 + d_1^\alpha } \right){\tau _1}}} - {e^{ - \left( {1 + d_1^\alpha } \right){\xi _1}}} + \frac{{1 + d_1^\alpha }}{{1 + d_1^\alpha  + {c_n}}}\left( {{e^{ - \left( {1 + d_1^\alpha  + {c_n}} \right){\xi _1}}} - {e^{ - \left( {1 + d_1^\alpha  + {c_n}} \right){\tau _1}}}} \right)} \right]}}{{{e^{ - \left( {1 + d_1^\alpha } \right){\tau _1}}}\left( {1 - \Delta \left( {1 - {e^{ - {\tau _1}{c_n}}}} \right)} \right)}}} \right]^k}\nonumber \\
  &\times {\left[ {1 - \left( {1 - \Delta \left( {1 - {e^{ - {\tau _1}{c_n}}}} \right)} \right){e^{ - \left[ {\left( {1 + d_1^\alpha } \right) + \left( {1 + d_2^\alpha } \right)} \right]{\tau _1}}}} \right]^{K - k}}{\left[ {\left( {1 - \Delta \left( {1 - {e^{ - {\tau _1}{c_n}}}} \right)} \right){e^{ - \left[ {\left( {1 + d_1^\alpha } \right) + \left( {1 + d_2^\alpha } \right)} \right]{\tau _1}}}} \right]^k}.
\end{align}
\hrulefill \vspace*{0pt}
\end{figure*}
where $\xi_{1} {\rm{ = }}\frac{{\gamma _{t{h_1}}^{HD}}}{{\rho {a_1}}}$ and $\gamma _{t{h_1}}^{HD}=2^{2{R_{D_{1}}}}-1$ with ${R_{D_{1}}}$ being the target rate of $D_{1}$.
\end{corollary}
\subsection{Benchmarks for SRS and TRS schemes}\label{Benchmarks for SRS and TRS}
%In this subsection, the random relay selection (RRS) scheme can be seen as a baseline for comparison purposes. In this case, the relay $R_{i}$ is selected randomly to help the BS transmitting the information. That is to say that the RRS scheme is regarded as the special case for SRS/TRS schemes with $K=1$, which is independent of the number of relays.
In this subsection, we consider the random relay selection (RRS) scheme as a benchmark for comparison purposes, where the relay $R_i$ is selected randomly to help the BS transmitting the information. Note that $R_i$ selected maybe not the optimal one for the NOMA RS schemes. In this case, the RRS scheme is capable of being regarded as the special case for SRS/TRS schemes with $K=1$, which is independent of the number of relays.
As such, for SRS scheme, the outage probability
of the RRS scheme for FD/HD NOMA can be easily approximated as
\begin{align}\label{OP derived for FD relay selection two benchmark}
 P_{RRS}^{FD,SRS} \approx & 1 - \left[ {1 - \frac{\pi }{{2N}}\sum\limits_{n = 1}^N {\sqrt {1 - \phi _n^2} }
\left( {{\phi _n} + 1} \right)} \right. \nonumber \\
&\left. { \times \left( {1 - \frac{{{e^{ - {c_n}\tau }}}}{{1 + \rho \tau {c_n}{\Omega _{{\rm{LI}}}}}}}
  \right)} \right] {e^{ - \left( {1 + d_1^\alpha } \right)\tau  - \left( {1 + d_2^\alpha } \right)\tau }},
\end{align}
and
\begin{align}\label{OP derived for HD relay selection two benchmark}
 P_{RRS}^{HD,SRS} \approx & 1 - \left[ {1 - \frac{\pi }{{2N}}\sum\limits_{n = 1}^N {\sqrt {1 - \phi _n^2} }
  \left( {{\phi _n} + 1} \right)} \right. \nonumber \\
 &\left. { \times \left( {1 - {e^{ - {\tau _1}{c_n}}}} \right)} \right]{e^{ - \left(
 {1 + d_1^\alpha } \right){\tau _1} - \left( {1 + d_2^\alpha } \right){\tau _1}}},
\end{align}
respectively. Similarly, for TRS scheme, the outage probability of RRS scheme for FD/HD NOMA can be obtained from
\eqref{the last OP expression for TS} and \eqref{OP expression for the TS relay selection in HD mode} by setting $K=1$, respectively.
%To simplify the writing mode, the benchmarks of TRS for FD/HD NOMA are no longer present in this subsection.
%In order to facilitate understanding the RRS scheme for FD/HD NOMA for the SRS scheme

\subsection{Diversity Order Analysis}
To gain more insights for these two RS schemes, the asymptotic diversity analysis is
provided in the high SNR region according to the derived outage probabilities. The diversity order is defined as
\begin{align}\label{diversity order}
d =  - \mathop {\lim }\limits_{\rho  \to \infty } \frac{{\log \left( {P^\infty \left( \rho  \right)} \right)}}{{\log \rho }},
\end{align}
where ${P^\infty \left( \rho  \right)}$ is the asymptotic outage probability.
\subsubsection{Single-stage Relay Selection Scheme}
Based on the analytical results in \eqref{OP derived for FD relay selection two}, when $\rho  \to \infty $,
we can derive the asymptotic outage probability of SRS scheme for FD NOMA in the following corollary.
\begin{corollary}\label{the asymptotic OP of SRS scheme for FD}
The asymptotic outage probability of FD-based NOMA SRS scheme at high SNR is given by
\begin{align}\label{asymptotic OP for the first relay selection FD mode}
P_{SRS}^{FD,\infty } = {\left[ {\frac{\pi }{{2N}}\sum\limits_{n = 1}^N {\sqrt {1 - \phi _n^2} } \left( {{\phi _n} + 1} \right)\left( {\frac{{\rho \tau {c_n}{\Omega _{{\rm{LI}}}}}}{{1 + \rho \tau{c_n} {\Omega _{{\rm{LI}}}}}}} \right)} \right]^K}.
\end{align}
\end{corollary}
Substituting \eqref{asymptotic OP for the first relay selection FD mode} into \eqref{diversity order}, we can obtain $d_{SRS}^{FD}= 0$.
\begin{remark}\label{remarks the first relay selection for FD}
The diversity order of SRS scheme for FD NOMA is zero, which is the same as the conventional FD RS
scheme.
\end{remark}

\begin{corollary}\label{asymptotic OP for HD SRS}
For the special case $\varpi=0$, the asymptotic outage probability
of HD-based NOMA SRS scheme with ${e^{-x}}\approx 1-x$ at high SNR is given by
\begin{align}\label{asymptotic OP for the first relay selection HD mode}
 P _{SRS}^{HD,\infty } =& \left[ {1 - \left( {1 - \frac{\pi }{{2N}}\sum\limits_{n = 1}^N {\sqrt {1 - \phi _n^2} } \left( {{\phi _n} + 1} \right){\tau _1}{c_n}} \right)} \right. \nonumber \\
 &{\left. { \times \left( {1 - \left( {1 + d_1^\alpha  + 1 + d_2^\alpha } \right){\tau _1}} \right)} \right]^K}.
\end{align}
\end{corollary}
Substituting \eqref{asymptotic OP for the first relay selection HD mode} into \eqref{diversity order}, we can obtain
$d_{SRS}^{HD}= K$.
\begin{remark}\label{remarks the first relay selection for HD}
The diversity order of SRS scheme for HD NOMA is $K$, which provides a diversity order equal to the number of
the available relays.
\end{remark}
\subsubsection{Two-stage Relay Selection Scheme}
As such, we can derive asymptotic outage probability of TRS scheme for FD NOMA in the following corollary.
\begin{corollary}\label{asymptotic OP for FD TRS}
The asymptotic outage probability of FD-based NOMA TRS scheme at high SNR is given by
\begin{align}\label{asymptotic OP for TS relay selection FD mode}
 P_{TRS}^{FD,\infty } =& {\sum\limits_{k = 0}^K {{K \choose
  k  }\left[ {\frac{{\Delta \left( {2\chi  - \psi } \right)}}{{1 - \Delta \left( {1 - \chi } \right)}}} \right]} ^k}{\left[ {\Delta \left( {1 - \chi } \right)} \right]^{K - k}}\nonumber \\
 & \times {\left[ {\left( {1 - \Delta \left( {1 - \chi } \right)} \right)} \right]^k}.
\end{align}
\begin{proof}
See Appendix~C.
\end{proof}
\end{corollary}

Upon substituting \eqref{asymptotic OP for TS relay selection FD mode} into \eqref{diversity order}, we obtain $d_{TRS}^{FD}= 0$.
\begin{remark}\label{remarks in TS relay selection for FD}
 The zero diversity order of TRS scheme for FD NOMA is obtained, which is the same as the FD-based SRS scheme.
\end{remark}

\begin{corollary}\label{asymptotic OP for HD TRS}
For the special case $\varpi=0$, the asymptotic outage probability of HD-based NOMA TRS scheme with ${e^{-x}}\approx 1-x$ at high SNR
is given by
\begin{align}\label{asymptotic OP for TS relay selection HD}
 P_{TSR}^{HD,\infty } =& \sum\limits_{k = 0}^K {{K \choose
  k  }\left[ {\frac{{\Delta {c_n}\left[ {{\xi _1} - {\tau _1}} \right]{\rm{ + }}\left( {1 + d_1^\alpha } \right){\xi _1}}}{{1 - \Delta {\tau _1}{c_n}}}} \right.} \nonumber \\
   &- \frac{{\left( {1 + d_1^\alpha } \right){\tau _1} + \Delta \left[ {\left( {1 + d_1^\alpha } \right){\xi _1} - \left( {1 + d_1^\alpha } \right){\tau _1}} \right]}}{{1 - \Delta {\tau _1}{c_n}}}\nonumber \\
 & {\left. { - \frac{{\Delta \left( {1 + d_1^\alpha } \right)\left( {{\tau _1} - {\xi _1} } \right)}}{{1 - \Delta {\tau _1}{c_n}}}} \right]^k}{\left( {\Delta {\tau _1}{c_n}} \right)^{K - k}}{\left( {1 - \Delta {\tau _1}{c_n}} \right)^k}.
\end{align}
\end{corollary}

Upon substituting \eqref{asymptotic OP for TS relay selection HD} into \eqref{diversity order}, we obtain $d_{TRS}^{HD}= K$.

\begin{remark}\label{remarks in TS relay selection for HD}
 The diversity order of TRS scheme for HD cooperative NOMA is K, which is the same as the HD-based SRS scheme.
\end{remark}
\subsubsection{Random Relay Selection Scheme}
For SRS scheme, based on the analytical results in \eqref{asymptotic OP for the first relay selection FD mode}
and \eqref{asymptotic OP for the first relay selection HD mode}, the
asymptotic outage probability of RRS scheme for FD/HD NOMA with $K=1$ can be given by
\begin{align}\label{RRS asymptotic OP for the SRS FD mode}
P_{RRS,SRS}^{FD,\infty } = \frac{\pi }{{2N}}\sum\limits_{n = 1}^N {\sqrt {1 - \phi _n^2} } \left( {{\phi _n} + 1} \right)\left( {\frac{{{c_n}\rho \tau {\Omega _{{\rm{LI}}}}}}{{1 + {c_n}\rho \tau {\Omega _{{\rm{LI}}}}}}} \right),
\end{align}
and
\begin{align}\label{RRS asymptotic OP for the SRS HD mode}
 P _{RRS,SRS}^{HD,\infty } =&  {1 - \left( {1 - \frac{\pi }{{2N}}\sum\limits_{n = 1}^N
 {\sqrt {1 - \phi _n^2} } \left( {{\phi _n} + 1} \right){\tau _1}{c_n}} \right)} \nonumber \\
 &{ { \times \left( {1 - \left( {1 + d_1^\alpha  + 1 + d_2^\alpha } \right){\tau _1}} \right)} },
\end{align}
respectively.

For TRS scheme, based on the analytical results in \eqref{asymptotic OP for TS relay selection FD mode}
and \eqref{asymptotic OP for TS relay selection HD}, the
asymptotic outage probability of RRS scheme for FD/HD NOMA with $K=1$ can be given by
\begin{align}\label{RRS asymptotic OP for the TRS  FD mode}
 P_{RRS,TRS}^{FD,\infty } =& {\sum\limits_{k = 0}^1 {{1 \choose
  k  }\left[ {\frac{{\Delta \left( {2\chi  - \psi } \right)}}{{1 - \Delta \left( {1 - \chi } \right)}}} \right]} ^k}{\left[ {\Delta \left( {1 - \chi } \right)} \right]^{1 - k}}\nonumber \\
 & \times {\left[ {\left( {1 - \Delta \left( {1 - \chi } \right)} \right)} \right]^k},
\end{align}
and
\begin{align}\label{RRS asymptotic OP for the TRS HD mode}
 P_{RRS,TSR}^{HD,\infty } =& \sum\limits_{k = 0}^1 {{1 \choose
  k  }\left[ {\frac{{\Delta {c_n}\left[ {{\xi _1} - {\tau _1}} \right]{\rm{ + }}\left( {1 + d_1^\alpha } \right){\xi _1}}}{{1 - \Delta {\tau _1}{c_n}}}} \right.} \nonumber \\
   &- \frac{{\left( {1 + d_1^\alpha } \right){\tau _1} + \Delta \left[ {\left( {1 + d_1^\alpha } \right){\xi _1} - \left( {1 + d_1^\alpha } \right){\tau _1}} \right]}}{{1 - \Delta {\tau _1}{c_n}}}\nonumber \\
 & {\left. { - \frac{{\Delta \left( {1 + d_1^\alpha } \right)\left( {{\tau _1} - \xi } \right)}}{{1 - \Delta {\tau _1}{c_n}}}}
  \right]^k}{\left( {\Delta {\tau _1}{c_n}} \right)^{1 - k}}\nonumber \\
  &\times {\left( {1 - \Delta {\tau _1}{c_n}} \right)^k},
\end{align}
respectively.
\begin{remark}\label{remarks for random relay selection}
Substituting \eqref{RRS asymptotic OP for the SRS FD mode}, \eqref{RRS asymptotic OP for the SRS HD mode}
and \eqref{RRS asymptotic OP for the TRS  FD mode}, \eqref{RRS asymptotic OP for the TRS HD mode} into
\eqref{diversity order}, we can observed that the diversity orders of RRS scheme for FD-NOMA and HD-NOMA are zero and one, respectively.
\end{remark}
%\textcolor[rgb]{0.00,0.00,1.00}{ As can be observed from the above analytical results, FD/HD-based SRS schemes inherit advantage to ensure the data rate of $D_2$, where the application of small packets can be achieved. The BS is capable of sending the critical safety information, such as blood pressure, pulse and heart rates. However, the merit of FD/HD-based TRS schemes is that in addition to guarantee the data rate of $D_2$, the BS can support $D_1$ to carry out some background tasks, i.e., downloading a movie or multimedia files.}
In order to get intuitional insights, as shown in TABLE~\ref{parameter1}, the diversity orders and application scenarios of FD/HD-based NOMA RS schemes are summarized to illustrate the comparison between them. For the sake of simplicity, we use ``D'' to represent the diversity order.
\begin{table}[!t]
\begin{center}
{\tabcolsep10pt
\begin{tabular}{|c|c|c|c|c}%表格中字体居中对齐
\cline{1-4}
\textbf{Duplex mode} & \textbf{RS scheme}  &\textbf{D}   &\textbf{Application scenario}\\
     \cline{1-4}
\multirow{3}{*}{FD NOMA}  & \multirow{1}{*}{SRS}   & 0   & Small packet service \\
\cline{2-4}
                     & \multirow{1}{*}{TRS}        & 0     & Background tasks\\
\cline{2-4}
                     & \multirow{1}{*}{RRS}        & 0    &   ------  \\
\cline{1-4}
\multirow{3}{*}{HD NOMA}  & \multirow{1}{*}{SRS}   & K      &  Small packet service  \\
\cline{2-4}
                       &\multirow{1}{*}{TRS}      &K        & Background tasks   \\
\cline{2-4}
                     & \multirow{1}{*}{RRS}       & 1      & ------  \\
\cline{1-4}
\end{tabular}}{}
\end{center}
\caption{Diversity orders and application scenarios for FD/HD-based NOMA RS schemes.}
\label{parameter1}
\end{table}

%\begin{table}[!t]
%\begin{center}
%{\tabcolsep14pt
%\begin{tabular}{|l|l|l|l}
%\cline{1-3}
%\textbf{Duplex mode} & \textbf{RS scheme}  &\textbf{D}   \\
%     \cline{1-3}
%\multirow{3}{*}{FD NOMA}    & \multirow{1}{*}{SRS}   & 0 \\
%\cline{2-3}
%                     & \multirow{1}{*}{TRS} & 0 \\
% \cline{2-3}
%                     & \multirow{1}{*}{RRS} & 0 \\
%\cline{1-3}
%\multirow{3}{*}{HD NOMA}  & \multirow{1}{*}{SRS} & K  \\
%\cline{2-3}
%               &\multirow{1}{*}{TRS}  &K &      \\
%\cline{2-3}
%                     & \multirow{1}{*}{RRS} & 1 \\
%\cline{1-3}
%\end{tabular}}{}
%\end{center}
%\caption{Diversity orders for FD/HD NOMA networks.}
%\label{parameter1}
%\end{table}

\subsection{Throughput Analysis}\label{Throughput Analysis}
In this subsection, the delay-limited transmission modes of these RS schemes are investigated for FD/HD NOMA networks.
In this mode, the BS sends information at a constant rate and the system throughput is subjective to the effect of outage probability.
Hence it is significant to discuss the system throughput for delay-limited mode in practical scenarios.
\begin{proposition}\label{proposition for throughput}
Based on above explanation, the system throughput of the RS schemes for FD/HD NOMA are given by
\begin{align}\label{Throughput Analysis for FD-based RS scheme}
R_\Psi ^{FD} = \left( {1 - P_\Psi ^{FD}} \right){R_{{D_1}}} + \left( {1 - P_\Psi ^{FD}} \right){R_{{D_2}}},
\end{align}
and
\begin{align}\label{Throughput Analysis for HD-based RS scheme}
R_\Psi ^{HD} = \left( {1 - P_\Psi ^{HD}} \right){R_{{D_1}}} + \left( {1 - P_\Psi ^{HD}} \right){R_{{D_2}}},
\end{align}
respectively, where $\Psi  \in \left\{ {SRS,TRS} \right\}$. ${R_{SRS}}$ and ${R_{TRS}}$ are system throughputs of
single-stage and two-stage RS schemes, respectively.
\end{proposition}
\section{Numerical Results}\label{Numerical Results}
In this section, our numerical results are provided for characterizing the outage performance of these two kinds of RS schemes. Monte Carlo simulation parameters used in this section are summarized in Table~\ref{parameter} \cite{DingTWC6779694,Liu2016TVT}, in which BPCU is short for bit per channel use. The complexity-vs-accuracy tradeoff parameter is $N=15$. Except FD/HD-based NOMA RRS schemes, the performance of OMA-based RS scheme is also shown as a benchmark for comparison, where the total communication process is finished in four slots. In the first slot, the BS sends information $x_{1}$ to relay $R_i$ and send $x_{2}$ to $R_i$ in the second slot. In the third and fourth slot, $R_i$ decodes and forwards the information  $x_{1}$ and $x_{2}$ to $D_1$ and $D_2$, respectively. Adding the performance of AF-based RS schemes for comparison will further enrich this paper, but this is beyond the scope of this paper. Note that NOMA users with low target data rate can be applied to the IoT scenarios, i.e., the low energy consumption, small packet service and so on.

\begin{table}[!t]
\centering
\tabcolsep4pt
\begin{tabular}{|l|l|}
\hline
Monte Carlo simulations repeated  &  ${10^6}$ iterations \\ \hline
Power allocation coefficients of NOMA&  $a_1=0.2$, $a_2=0.8$   \\ \hline
Targeted data rates & $R_{D_{1}}=1 $, $R_{D_{2}}=0.1$ BPCU  \\ \hline
Pass loss exponent  & $\alpha=2$  \\ \hline
The radius of a disc region &  ${R_{\cal D}}=2$ m \\ \hline
The distance between the BS and $D_{1}$ &  $10$ m \\ \hline
The distance between the BS and $D_{2}$ & $12$ m \\ \hline
\end{tabular}
\caption{Table of Parameters for numerical results.}
\label{parameter}
\end{table}
\subsection{Single-stage Relay Selection Scheme}
In this subsection, the FD/HD-based NOMA RRS schemes and OMA-based RS schemes are regarded as the baselines for comparison purposes.

\begin{figure}[t!]
    \begin{center}
        \includegraphics[width=3.48in,  height=2.8in]{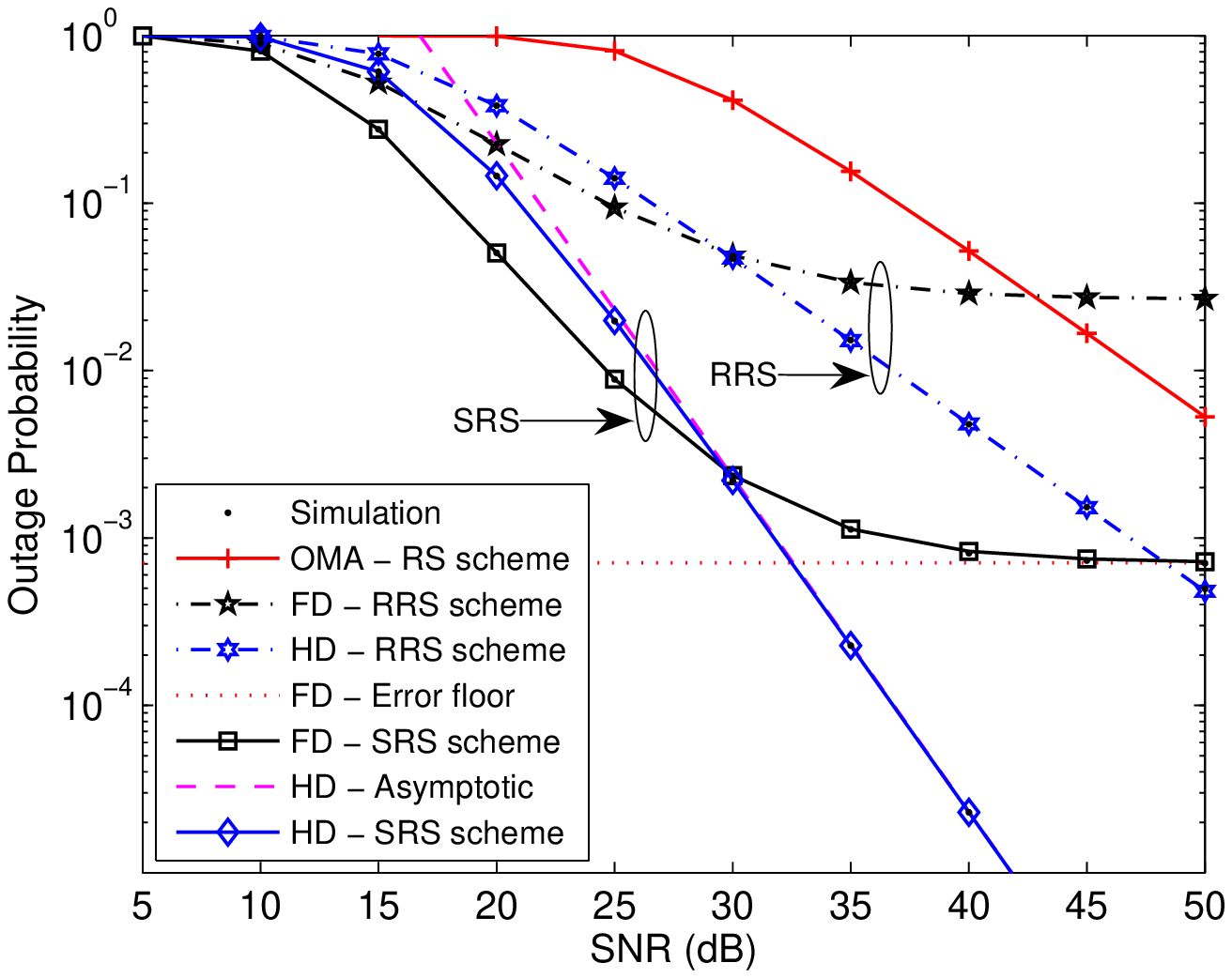}
        \caption{Outage probability versus the transmit SNR for SRS scheme with $K=2$; ${R_{D_{1}}}=1$, ${R_{D_{2}}}=0.1$ BPCU and $\mathbb{E}\{|h_{LI}|^2\}=-10$ dB.}
        \label{Selection_One_OP_R2}
    \end{center}
\end{figure}
\begin{figure}[t!]
    \begin{center}
        \includegraphics[width=3.48in,  height=2.8in]{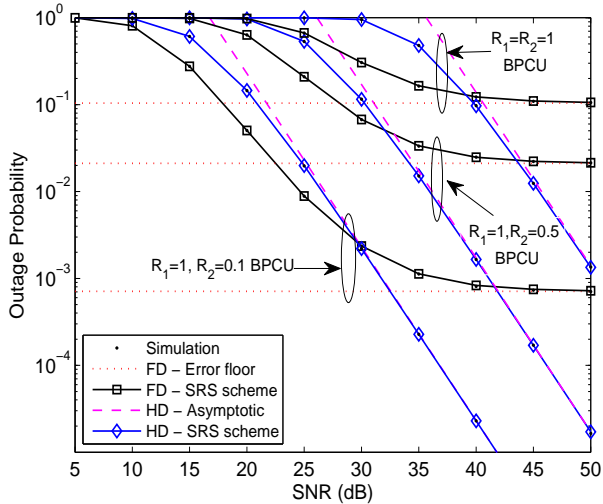}
        \caption{Outage probability versus the transmit SNR for SRS scheme with the different target rates; $K=2$ and $\mathbb{E}\{|h_{LI}|^2\}=-10$ dB.}
        \label{Selection_One_OP_R2_diff_target_rate}
    \end{center}
\end{figure}

Fig. \ref{Selection_One_OP_R2} plots the outage probability of SRS scheme versus SNR for a simulation setting with $K=2$ and $\mathbb{E}\{|h_{LI}|^2\}=-10$ dB. The black and blue solid curves are the SRS scheme for FD/HD NOMA, corresponding to the approximate results derived in \eqref{OP derived for FD relay selection two} and \eqref{OP expression for the second HD relay selection}, respectively. The dash dotted curves represent the approximate outage probabilities of RRS schemes for FD/HD NOMA derived in \eqref{OP derived for FD relay selection two benchmark} and \eqref{OP derived for HD relay selection two benchmark}, respectively. Obviously, the outage probability curves match precisely with the Monte Carlo simulation results.
It is observed that the performance of FD-based NOMA SRS scheme is superior to HD-based NOMA scheme on the condition of low SNR region. The reason is that loop interference is not the dominant impact factor for FD cooperative NOMA in the low SNR region.
Moreover, the outage performance of the HD-based NOMA SRS scheme outperforms the HD-based RRS scheme. Another observation is that HD-based NOMA SRS scheme is superior to OMA-based RS scheme. This is due to the fact that HD-based NOMA RS schemes is capable of enhancing the spectral efficiency compared to OMA.
The asymptotic outage probability cures of the SRS schemes for FD/HD NOMA are plotted according to the analytical results in \eqref{asymptotic OP for the first relay selection FD mode} and \eqref{asymptotic OP for the first relay selection HD mode}, respectively. One can observe that the asymptotic curves well approximate the analytical performance curves in the high SNR region. It is worth noting that an error floor exists in the FD-based NOMA SRS scheme, which verifies the conclusion in \textbf{Remark \ref{remarks the first relay selection for FD}} and obtain zero diversity order. This is due to the fact that there is the loop interference in FD NOMA.

\begin{figure}[t!]
    \begin{center}
        \includegraphics[width=3.48in,  height=2.8in]{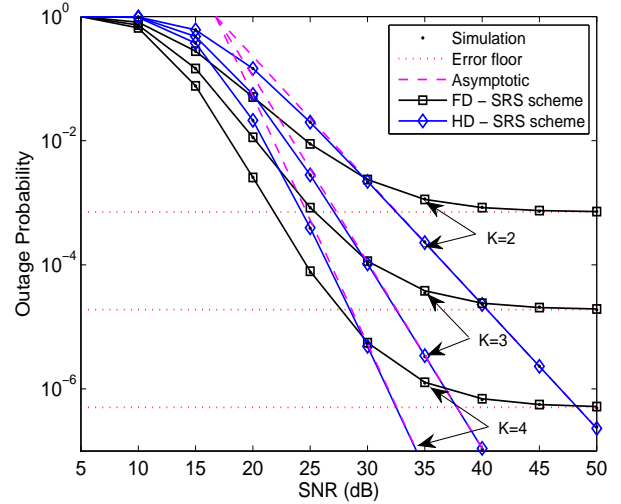}
        \caption{Outage probability versus the transmit SNR for SRS scheme with $K = 2, 3, 4$; ${R_{D_{1}}}=1$, ${R_{D_{2}}}=0.1$ BPCU and $\mathbb{E}\{|h_{LI}|^2\}=-10$  dB.}
        \label{Selection_One_OP_Ralays_234}
    \end{center}
\end{figure}
\begin{figure}[t!]
    \begin{center}
        \includegraphics[width=3.48in,  height=2.8in]{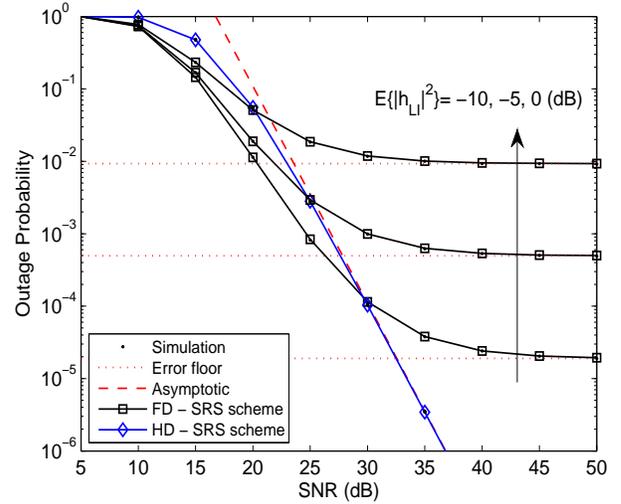}
        \caption{Outage probability versus the transmit SNR for SRS scheme with $K = 3$; ${R_{D_{1}}}=1$, ${R_{D_{2}}}=0.1$ BPCU.}
        \label{Selection_One_OP_Ralay3_LI}
    \end{center}
\end{figure}

Fig. \ref{Selection_One_OP_R2_diff_target_rate} plots the outage probability of SRS scheme with different target rates. One can observe that adjusting the target rates of NOMA users will affect the outage behaviors of the FD/HD-based SRS schemes. As the value of target rates increases, the superior of FD/HD-based NOMA SRS schemes becomes not obvious. It is worth noting that based on the application requirements of different scenarios, the setting of reasonable target rates for NOMA users is prerequisite.

Fig. \ref{Selection_One_OP_Ralays_234} plots the outage probability of SRS scheme versus SNR for a simulation setting with $K=2,3,4$ relays and $\mathbb{E}\{|h_{LI}|^2\}=-10$ dB. As can be seen that the analytical curves perfectly match with the simulation results, while the approximations match the analytical performance curves in the high SNR region. It is shown
that the number of relays in the networks considered strongly affect the performance of FD/HD-based NOMA SRS schemes. With the
number of relays increasing, the lower outage probability are achieved by this RS scheme.
This is because more relays bring higher diversity gains, which improves the reliability of the cooperative networks.
Another observation is that the HD-based NOMA SRS scheme provides a diversity order that is equal to the number of the relays ($K$), which verifies the conclusion in \textbf{Remark \ref{remarks the first relay selection for HD}}.
As a further development,
Fig. \ref{Selection_One_OP_Ralay3_LI} plots the outage probability of SRS scheme versus different values of LI
from $\mathbb{E}\{|h_{LI}|^2\}=-10$ dB to $\mathbb{E}\{|h_{LI}|^2\}=5$ dB.
As observed from the figure, we can see that the value of LI also strongly affect the performance of FD-based SRS scheme for
NOMA, while the HD-based SRS scheme is not affected. This is due to the fact that LI is not existent for the HD-based SRS scheme with $\varpi=0$.
As the value of LI becomes larger, the outage performance of the FD-based SRS scheme becomes more worse. In consequence, it is significant to consider the influence of LI in the practical FD NOMA networks.

\begin{figure}[t!]
    \begin{center}
        \includegraphics[width=3.48in,  height=2.8in]{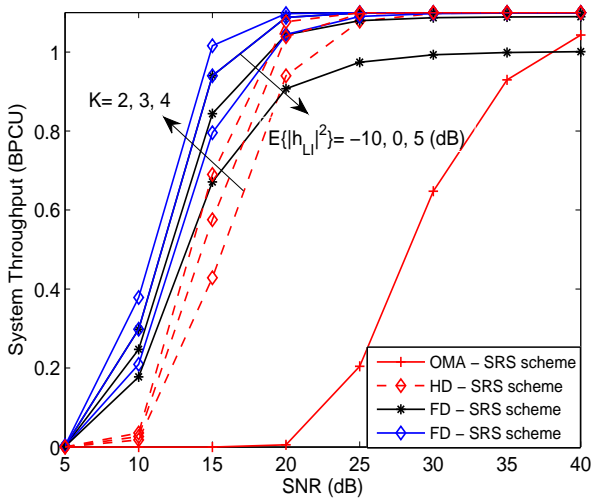}
        \caption{System throughput in delay-limited transmission mode versus SNR for the SRS scheme.}
        \label{SRS_throughput_Ralay3_LI}
    \end{center}
\end{figure}
\begin{figure}[t!]
    \begin{center}
        \includegraphics[width=3.48in,  height=2.8in]{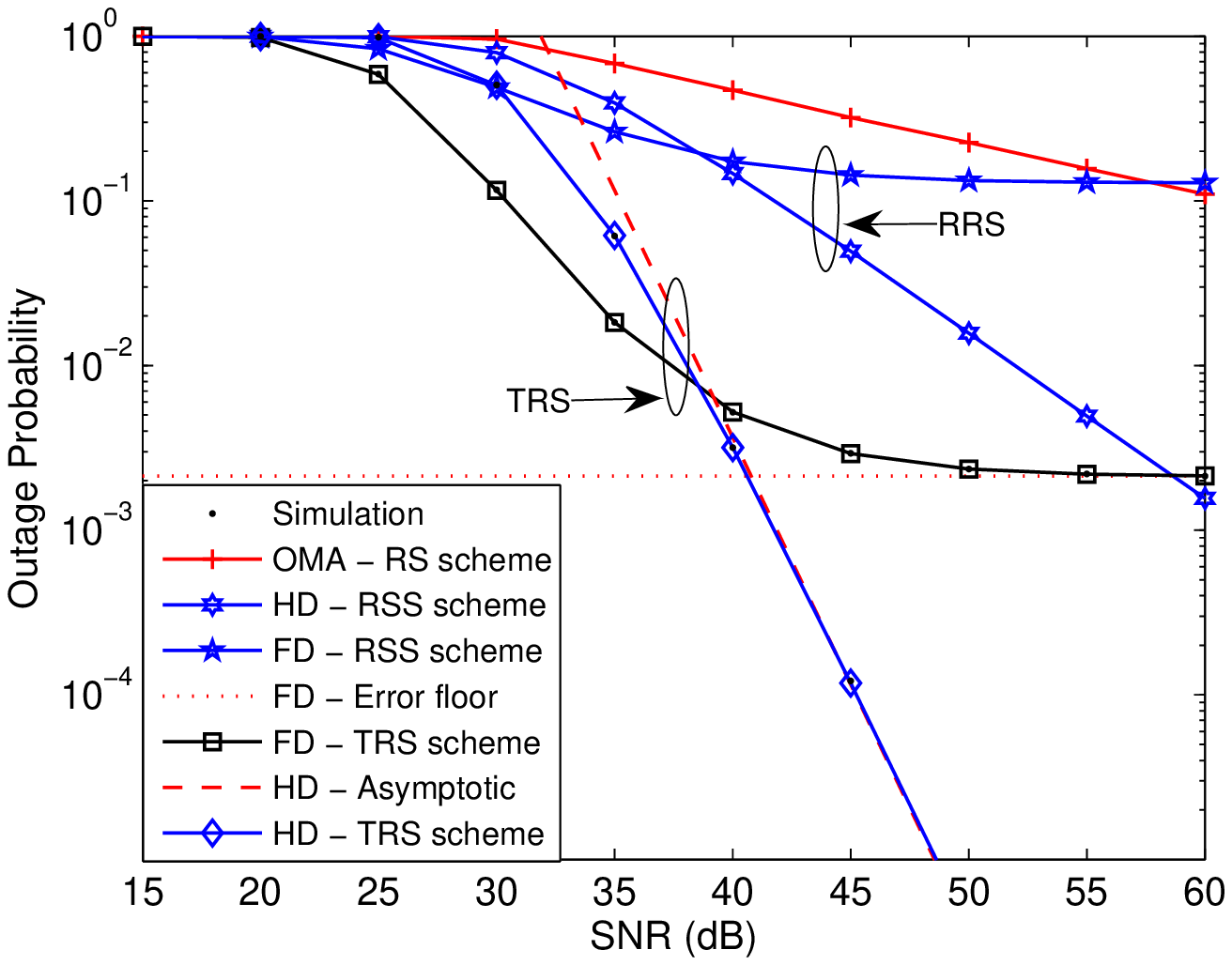}
        \caption{Outage probability versus the transmit SNR for TRS scheme with $K=3$; ${R_{D_{1}}}=1$, ${R_{D_{2}}}=0.1$ BPCU and $\mathbb{E}\{|h_{LI}|^2\}=-20$ dB.}
        \label{Two_stage_RS_OP_R3_add_OMA}
    \end{center}
\end{figure}
\begin{figure}[t!]
    \begin{center}
        \includegraphics[width=3.48in,  height=2.8in]{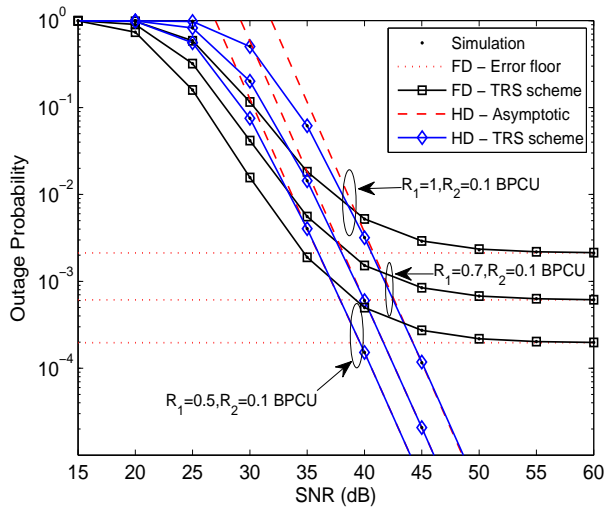}
        \caption{Outage probability versus the transmit SNR for TRS scheme with different target rates; $K=3$ and $\mathbb{E}\{|h_{LI}|^2\}=-20$ dB.}
        \label{The_TRS_OP_add_diff_target_rate}
    \end{center}
\end{figure}

Fig. \ref{SRS_throughput_Ralay3_LI} plots system throughput versus SNR in delay-limited transmission mode for the different number of relays from $K=2$ to $K=4$ with $\mathbb{E}\{|h_{LI}|^2\}=-10$ dB. The blue solid and red dashed curves
represent throughput of SRS scheme for FD/HD NOMA networks which are obtained from \eqref{Throughput Analysis for FD-based RS scheme}
and \eqref{Throughput Analysis for HD-based RS scheme}, respectively. One can observe that the FD-based SRS scheme achieves
a higher throughput compared to the HD-based SRS scheme for NOMA networks. This is because that the value of LI
has a smaller influence for the outage behavior of FD NOMA in the low SNR region. Furthermore, the FD/HD-based NOMA SRS schemes outperform OMA-based RS scheme in terms of system throughput. This is due to the fact that NOMA-based SRS scheme can provide more spectrum efficiency than OMA-based SRS scheme.
As the number of relays becomes larger, the FD/HD-based SRS schemes can improve the system throughput. This phenomenon can be explained as that a lower outage probability can be obtained by the FD/HD-based SRS schemes.
In addition, Fig. \ref{SRS_throughput_Ralay3_LI} further give system throughput in delay-limited transmission mode for the different
values of LI with $K=3$. As can be observed that increasing the values of LI from $\mathbb{E}\{|h_{LI}|^2\}=-10$ dB to
$\mathbb{E}\{|h_{LI}|^2\}=5$ dB reduces the system throughput. This phenomenon indicates that it is of significance to consider
the impact of LI for FD-based SRS scheme when designing practical cooperative NOMA systems.

\begin{figure}[t!]
    \begin{center}
        \includegraphics[width=3.48in,  height=2.8in]{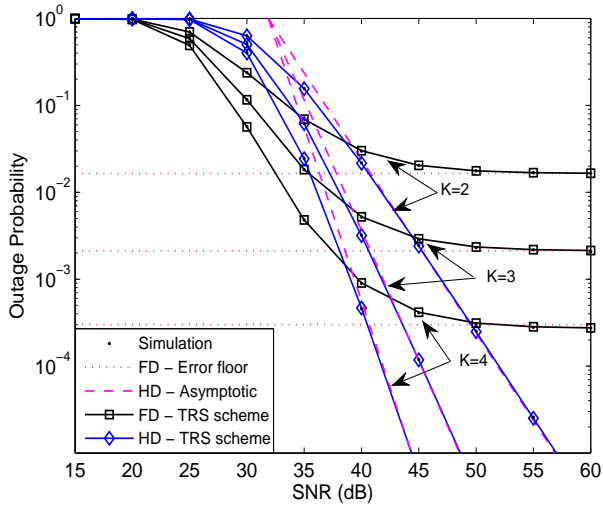}%   10_nodirect_HD_FD_5DB_new
        \caption{Outage probability versus the transmit SNR for TRS scheme with $K = 2, 3, 4$; ${R_{D_{1}}}=1$, ${R_{D_{2}}}=0.1$ BPCU and $\mathbb{E}\{|h_{LI}|^2\}=-20$ dB.}
        \label{Two_stage_RS_OP_Relays_234}
    \end{center}
\end{figure}

\subsection{Two-stage Relay Selection Scheme}
In this subsection, except FD/HD-based NOMA RRS scheme, the outage performance of OMA-based RS scheme is also shown as a benchmark for comparison.

\begin{figure}[t!]
    \begin{center}
        \includegraphics[width=3.48in,  height=2.8in]{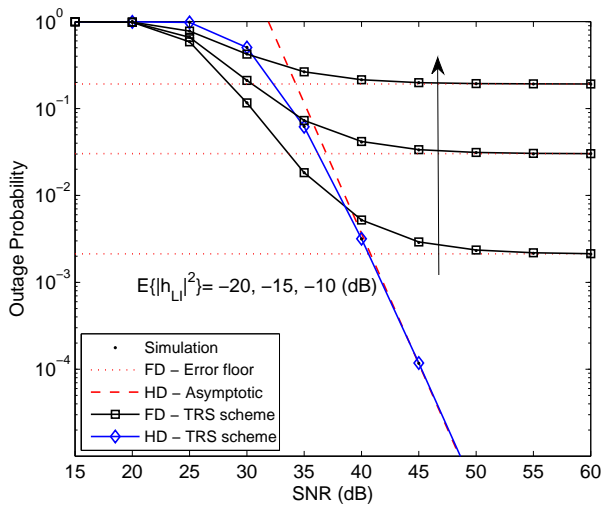}%   10_nodirect_HD_FD_5DB_new
        \caption{Outage probability versus the transmit SNR for TRS scheme with $K = 3$; ${R_{D_{1}}}=1$, ${R_{D_{2}}}=0.1$ BPCU.}
        \label{Two_stage_RS_relay3_LI}
    \end{center}
\end{figure}
\begin{figure}[t!]
    \begin{center}
        \includegraphics[width=3.48in,  height=2.8in]{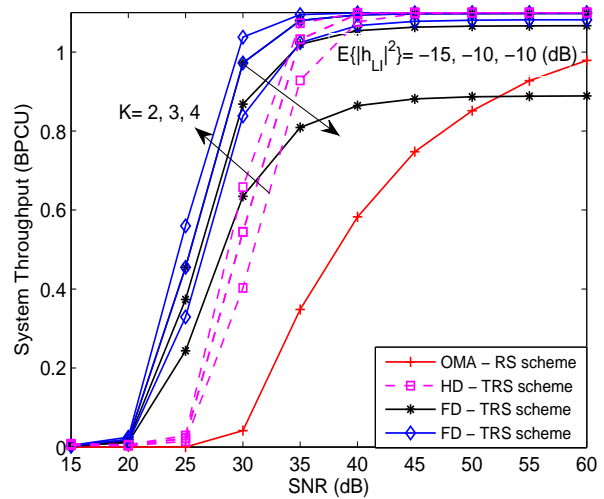}%   10_nodirect_HD_FD_5DB_new
          \caption{System throughput in delay-limited transmission mode versus SNR for the TRS scheme with ${R_{D_{1}}}=1$, ${R_{D_{2}}}=0.1$ BPCU.}
        \label{TRS_throughput_relays_LI}
    \end{center}
\end{figure}

Fig. \ref{Two_stage_RS_OP_R3_add_OMA} plots the outage probability of TRS scheme versus SNR with setting to be $K=3$ and
$\mathbb{E}\{|h_{LI}|^2\}=-20$ dB. The approximate analytical curves of the TRS schemes for FD/HD NOMA are plotted based
on \eqref{the last OP expression for TS} and \eqref{OP expression for the TS relay selection in HD mode}, respectively. As can
be observed from the figure, the analytical curves match perfectly with Monte Carlo simulation results.
We confirm that the higher outage performance can be obtained by FD-based NOMA TRS scheme in the low SNR region. This is due to fact that there is a low loop interference for FD-based TRS scheme and does not suffer from bandwidth-loss influence.
One can observe that the outage behaviors of FD/HD-based NOMA TRS schemes outperform OMA-based RS scheme.
The asymptotic outage probability curves of FD/HD-based NOMA TRS scheme are plotted
according to \eqref{asymptotic OP for TS relay selection FD mode} and \eqref{asymptotic OP for TS relay selection HD},
which are practically indistinguishable from the analytical results. It is also observed that the FD-based TRS scheme
for NOMA converges to an error floor and obtains the zero diversity, which is due to the fact that the loop interference
exists at the relay. This phenomenon is confirmed by the insights in \textbf{Remark \ref{remarks in TS relay selection for FD}}.  However, the HD-based TRS scheme for NOMA overcomes the problem of zero diversity inherent to FD-based scheme.

As shown in Fig. \ref{Selection_One_OP_R2_diff_target_rate},
Fig. \ref{The_TRS_OP_add_diff_target_rate} plots the outage probability of TRS scheme with different target rates. It is shown that when the target rates of NOMA users is reduced, the FD/HD-based NOMA TRS schemes is capable of providing better outage performance. We confirm that the IoT scenarios (i.e., small packet service) considered can be supported by the NOMA-based RS schemes.

Fig. \ref{Two_stage_RS_OP_Relays_234} plots the outage probability of TRS scheme versus SNR for a
simulation setting to be $K=2,3,4$ relays and $\mathbb{E}\{|h_{LI}|^2\}=-20$ dB. We observed that the number of
relays affect the performance of TRS scheme. With the number of relays increasing, the superiority of FD/HD-based NOMA TRS schemes is apparent and the lower outage probabilities are obtained. We also see that
the HD-based RS scheme is capable of achieving a diversity order of $K$, which confirms the insights in
\textbf{Remark \ref{remarks in TS relay selection for HD}}.
From a practical perspective, it is important to consider multiple relays in the networks when designing the NOMA
RS systems. Fig. \ref{Two_stage_RS_relay3_LI} plots the outage probability of the TRS scheme versus different values of LI from $\mathbb{E}\{|h_{LI}|^2\}=-20$ dB to $\mathbb{E}\{|h_{LI}|^2\}=-10$ dB.
We also can observe that with the value of LI increasing, the superior of outage performance for the FD-based TRS scheme
is not existent.

Fig. \ref{TRS_throughput_relays_LI} plots system
throughput versus SNR in delay-limited transmission mode for the different number of relays from $K=2$ to
$K=4$ with $\mathbb{E}\{|h_{LI}|^2\}=-20$ dB. The solid black and dashed magenta curves
represent throughput of TRS for FD/HD NOMA networks which are obtained from \eqref{Throughput Analysis for FD-based RS scheme}
and \eqref{Throughput Analysis for HD-based RS scheme}, respectively. We can also observe that FD-based NOMA TRS scheme has a higher throughput than HD-based scheme in the low SNR region. The reason is that the FD-based TRS scheme is capable of achieving a lower outage probability compared to HD-based scheme. Moreover, the throughput of
FD/HD-NOMA TRS schemes precedes that of OMA-based RS scheme. Additionally, it is worth pointing out that adjusting the size of target data rate (i.e., $R_{D_1}$ and $R_{D_2}$) will affect the system throughput for delay-limited transmission mode.
The main performance of TRS scheme trends follow those in Fig. \ref{SRS_throughput_Ralay3_LI}.
Additionally, as can be seen from the figure that increasing the values of LI from $\mathbb{E}\{|h_{LI}|^2\}=-20$ dB to
$\mathbb{E}\{|h_{LI}|^2\}=-10$ dB reduces the system throughput and the existence of the throughput ceilings in the high SNR region.
This is due to the fact that the FD-based TRS scheme converges to the error floor.

\section{Conclusions}\label{Conclusions}
This paper has investigated a pair of RS schemes for FD/HD NOMA networks insightfully. Stochastic geometry has been employed for modeling the locations of relays in the network.
New analytical expressions of outage probability for two RS schemes have been derived. Due to the influence of LI at relay, a zero diversity order has been obtained by these two RS schemes for FD NOMA.  Based on the analytical results, it was demonstrated that the diversity orders of HD-based RS schemes were determined by the number of relays in the networks considered. Simulation results showed that the FD/HD-based NOMA SRS/TRS schemes are capable of providing better outage behaviors than RRS and OMA-based RS schemes. The system throughput of delay-limited transmission mode for FD/HD-based NOMA RS schemes were discussed.
The setting of perfect SIC operation my bring about overestimated performance for the RS schemes considered, hence our future treaties will consider the impact of imperfect SIC. Another promising future research direction is to optimize the power allocation between NOMA users, which can further enhance the performance of NOMA-based RS schemes.

\appendices
\section*{Appendix~A: Proof of Theorem \ref{theorem relay selection one}} \label{Appendix:A}
\renewcommand{\theequation}{A.\arabic{equation}}
\setcounter{equation}{0}

Let ${W_i} = \min \left\{ {{\gamma _{{D_2} \to {R_i}}},{\gamma _{{D_2} \to {D_1}}},{\gamma _{{D_2}}}} \right\}$, $W = \max \left\{ {{W_1},{W_2}...,{W_N}} \right\}$, then
\begin{align}\label{derived process for the second RS}
 \Pr \left( {W < \gamma _{t{h_2}}^{FD}} \right) =& \Pr \left( {\max \left\{ {{W_1},{W_2},...,{W_N}} \right\} < \gamma _{t{h_2}}^{FD}} \right) \nonumber\\
  =& \prod\limits_{i = 1}^K {\Pr\left( {{W_i} < {\gamma _{t{h_2}}^{FD}}} \right)}.
  %=& {\left[ {\Pr\left( {{W_i} < \gamma _{t{h_2}}^{FD}} \right)} \right]^K}.
\end{align}
Hence the outage probability of the FD-based SRS scheme only requires ${\Pr\left( {{W_i} < \gamma _{t{h_2}}^{FD}} \right)}$, which is given by
\begin{align}\label{additional derived process for the second RS}
 \Pr \left( {{W_i} < \gamma _{t{h_2}}^{FD}} \right) =& \Pr\left( {\min \left\{ {{\gamma _{{D_2} \to {R_i}}},{\gamma _{{D_2} \to {D_1}}},{\gamma _{{D_2}}}} \right\} < \gamma _{t{h_2}}^{FD}} \right) \nonumber\\
  =& 1 - \underbrace {\Pr \left( {{\gamma _{{D_2} \to {R_i}}} > \gamma _{t{h_2}}^{FD}} \right)}_{{J_{11}}} \nonumber\\
 & \times \underbrace {\Pr \left( {{\gamma _{{D_2} \to {D_1}}} > \gamma _{t{h_2}}^{FD}} \right)}_{J_{12}}\underbrace {\Pr \left( {{\gamma _{{D_2}}} > \gamma _{t{h_2}}^{FD}} \right)}_{{{J_{13}}}},
\end{align}
where $\varpi  = 1$.

Define ${X_i}{\rm{ = }}\frac{{{{\left| {{h_{S{R_i}}}} \right|}^2}}}{{1 + d_{S{R_i}}^\alpha }}$,
${Y_{1i}}{\rm{ = }}\frac{{{{\left| {{h_{{R_i}{D_1}}}} \right|}^2}}}{{1 + d_{{R_iD_1}}^\alpha }}$, ${Y_{2i}}{\rm{ = }}\frac{{{{\left| {{h_{{R_i}{D_2}}}} \right|}^2}}}{{1 + d_{{R_iD_2}}^\alpha }}$, and ${Z} = {\left| {{h_{LI}}} \right|^2}$.
As stated in \cite{Ding6868214,Liu7445146SWIPT} and by utilizing the polar coordinate, the CDF ${F_{{X_i}}}$ of $X_{i}$ is given by
\begin{align}\label{CDF of Xi}
{F_{{X_i}}}\left( x \right) = \frac{2}{{R_{\cal D}^2}}\int_0^{{R_{\cal D}}} {\left( {1 - {e^{ - \left( {1 + {r^\alpha }} \right)x}}} \right)} rdr.
\end{align}
%From the above equation we can observe that ,
However, for many communication scenarios $\alpha>2$, \eqref{CDF of Xi} does not have a closed-form solution. In this case, the approximate expression of \eqref{CDF of Xi} can be obtained by using Gaussian-Chebyshev quadrature \cite{Numerical1987} and given by
\begin{align}\label{approximation expression for CDF}
{F_{{X_i}}}\left( x \right) \approx \frac{\pi }{{2N}}\sum\limits_{n = 1}^N {\sqrt {1 - \phi _n^2} } \left( {1 - {e^{ - {c_n}x}}} \right)\left( {{\phi _n} + 1} \right).
\end{align}

Substituting \eqref{the SINR1 for relay} into \eqref{additional derived process for the second RS} and applying algebraic manipulations,  ${J_{11}}$ can be further expressed as follows:
\begin{align}\label{the expression of J_11}
{J_{11}}=& 1- \Pr \left( {{X_i} < \left( {\rho \varpi {f_Z}\left( z \right) + 1} \right){\tau }} \right),
\end{align}
where ${f_Z}\left( z \right) = \frac{1}{{{\Omega _{LI}}}}{e^{ - \frac{z}{{{\Omega _{LI}}}}}}$.
By the virtue of approximate expression of CDF for ${X_i}$ in \eqref{approximation expression for CDF},
${J_{11}}$ is calculated as
\begin{align}\label{J_11}
{J_{11}}=& {\rm{1}} - \int_0^\infty  {\frac{1}{{{\Omega _{{\rm{LI}}}}}}{e^{ - \frac{z}{{{\Omega _{{\rm{LI}}}}}}}}} {F_{{X_i}}}\left( {\left( {\rho \varpi z + 1} \right){\tau }} \right)dz\nonumber \\
 \approx& 1 - \frac{\pi }{{2N}}\sum\limits_{n = 1}^N {\sqrt {1 - \phi _n^2} } \left( {{\phi _n} + 1} \right)\left( {1 - \frac{{{e^{ - {c_n}{\tau }}}}}{{{1 + \varpi \rho  \tau {c_n}{\Omega _{{\rm{LI}}}}}}}} \right).
\end{align}

On the condition of ${d_{{R_iD_j}}} = \sqrt {d_{S{R_i}}^2 +d_j^2- 2{d_j}{d_{S{R_i}}}\cos \left( {{\theta _i}} \right)} $ and ${d_{{R_iD_j}}} \gg d_{SR_{i}}$, $j \in \left( {1,2} \right)$,
to further simplify computational complexity, we assume that the distance between $R_i$ and ${D_j}$ can be approximated as the distance between the BS and $D_j$, i.e., ${d_{{R_iD_j}}} \approx d_j$.
It is worth noting through this approximation, the distance $d_j$ between the BS and $D_j$ is a fixed value. Hence we can obtain the corresponding approximate CDF of ${F_{{Y_{ji}}}}$ i.e., ${F_{{Y_{ji}}}} = 1 - {e^{ - \left( {1 + d_j^\alpha } \right)\tau }}$.
Upon substituting \eqref{the SINR1 for D1 to detect D2} and \eqref{the SINR3 for D2} into \eqref{additional derived process for the second RS}, ${J_{12}}$ and ${J_{13}}$ are approximated by
\begin{align}\label{J_12}
{J_{12}} = \Pr \left( {{Y_{1i}} > {\tau }} \right) \approx {e^{ - \left( {1 + d_1^\alpha } \right){\tau }}},
\end{align}
and
\begin{align}\label{J_13}
{J_{13}} =\Pr \left( {{Y_{2i}} > {\tau }} \right) \approx {e^{ - \left( {1 + d_2^\alpha } \right){\tau }}},
\end{align}
respectively.
Combining \eqref{J_11}, \eqref{J_12}, and \eqref{J_13}, we can calculate ${\Pr\left( {{W_i} < \gamma _{t{h_2}}^{FD}} \right)}$. Finally, substituting \eqref{additional derived process for the second RS} into \eqref{derived process for the second RS}, we can obtain \eqref{OP derived for FD relay selection two}.
The proof is completed.

\appendices
\section*{Appendix~B: Proof of Lemma \ref{lemma:1}}\label{Appendix:B}
\renewcommand{\theequation}{B.\arabic{equation}}
\setcounter{equation}{0}

Based on \eqref{TS OP expression}, the conditional probability ${\Theta _1}$ can be expressed as
\begin{align*}\label{derived process for Theta1 in lemma}
 {\Theta _1}=&\Pr \left( {{s_i} < {R_{{D_1}}}{\rm{|}}  {\left| {S_R^2} \right| = k}} \right) \nonumber\\
 =& \Pr \left( {\min \left\{ {{\gamma _{{D_1} \to {R_i}}},{\gamma _{{D_1}}}} \right\} < \gamma _{t{h_1}}^{FD}{\rm{|}}i \in \left| {S_R^{2}} \right|,\left| {S_R^{2}} \right| > 0} \right) \nonumber\\
  =& \Pr \left( {{\gamma _{{D_1} \to {R_i}}} < {\gamma _{{D_1}}},{\gamma _{{D_1} \to {R_i}}} < \gamma _{t{h_1}}^{FD}} \right. \nonumber\\
 &\underbrace {\begin{array}{*{20}{c}}
   {} & {} & {} & {} & {\left. {{\rm{|}}{\gamma _{{D_2} \to {R_i}}} > \gamma _{t{h_2}}^{FD},{\gamma _{{D_2} \to {D_1}}} > \gamma _{t{h_2}}^{FD}} \right)} \nonumber \\
\end{array}}_{{J_{21}}}\nonumber \\
\end{align*}
\begin{align}
  &+ \Pr \left( {{\gamma _{{D_1}}} < {\gamma _{{D_1} \to {R_i}}},{\gamma _{{D_1}}} < \gamma _{t{h_1}}^{FD}} \right.\nonumber \\
 &\underbrace {\begin{array}{*{20}{c}}
   {} & {} & {} & {} & {\left. {{\rm{|}}{\gamma _{{D_2} \to {R_i}}} > \gamma _{t{h_2}}^{FD},{\gamma _{{D_2} \to {D_1}}} > \gamma _{t{h_2}}^{FD}} \right)}
\end{array}}_{{J_{31}}}
\end{align}
where $\varpi  = 1$ and $\gamma _{t{h_1}}^{FD}=2^{R_{D_{1}}}-1$ with $R_{D_{1}}$ being the target rate of $D_{1}$.

According to the definition of conditional probability, $J_{21}$ can be expressed as
\begin{align}\label{derived process for J1 in lemma}
{J_{21}} = \frac{\begin{array}{l}
 \Pr \left( {{\gamma _{{D_1} \to {R_i}}} < {\gamma _{{D_1}}},{\gamma _{{D_1} \to {R_i}}} < \gamma _{t{h_1}}^{FD}} \right. \\
 \begin{array}{*{20}{c}}
   {} & {} & {} & {} & {\left. {{\gamma _{{D_2} \to {R_i}}} > \gamma _{t{h_2}}^{FD},{\gamma _{{D_2} \to {D_1}}} > \gamma _{t{h_2}}^{FD}} \right)}  \\
\end{array} \\
 \end{array}}{{\Pr \left( {{\gamma _{{D_2} \to {R_i}}} > \gamma _{t{h_2}}^{FD},{\gamma _{{D_2} \to {D_1}}} > \gamma _{t{h_2}}^{FD}} \right)}}.
\end{align}
Define the numerator and denominator of $J_{21}$ in \eqref{derived process for J1 in lemma} by $\Xi_{1} $ and $\Xi_{2} $, respectively. Substituting \eqref{the SINR1 for relay}, \eqref{the SINR2 for relay}, \eqref{the SINR1 for D1 to detect D2} and
\eqref{the SINR3 for D2} to \eqref{derived process for J1 in lemma} and applying some algebraic manipulations, we rewrite $\Xi_{1} $ as follows:
\begin{align}\label{the expression of Xi}
 {\Xi _1} =& \Pr \left( {{X_i} < {Y_{1i}}\left( {\rho \varpi {Z} + 1} \right),{X_i} < \xi \left( {\rho \varpi {Z} + 1} \right),} \right. \nonumber\\
 &\begin{array}{*{20}{c}}
   {} & {} & {} & {} & {} & {} & {\left. {{X_i} > {\tau }\left( {\rho \varpi {Z} + 1} \right),{Y_{1i}} > {\tau }} \right)}  \nonumber\\
\end{array}\nonumber \\
  =& \Pr \left( {{\tau }\left( {\rho \varpi {Z} + 1} \right) < {X_i} < \xi \left( {\rho \varpi {Z} + 1} \right),{Y_{1i}} > \theta } \right)\nonumber \\
 & + \Pr \left( {{\tau }\left( {\rho \varpi {Z} + 1} \right) < {X_i} < {Y_{1i}}\left( {\rho \varpi {Z} + 1} \right),{\tau } < {Y_{1i}} < \xi } \right)\nonumber \\
  =& \int_0^\infty  {{f_{{Z}}}\left( z \right)} \int_\theta ^\infty  {{f_{{Y_{1i}}}}\left( y \right)} \left[ {{F_{{X_i}}}\left( {\xi \left( {\rho \varpi z + 1} \right)} \right)} \right. \nonumber \\
 &\underbrace {\begin{array}{*{20}{c}}
   {} & {} & {} & {} & {} & {} & {} &{} &{} &{} &{\left. { - {F_{{X_i}}}\left( {{\tau }\left( {\rho \varpi z + 1} \right)} \right)} \right]dydz} \nonumber \\
\end{array}}_{{J_{22}}} \nonumber \\
  &+ \int_0^\infty  {{f_{{Z}}}\left( z \right)} \int_{{\tau }}^\xi  {{f_{{Y_{1i}}}}\left( y \right)} \left[ {{F_{{X_i}}}\left( {y\left( {\rho \varpi z + 1} \right)} \right)} \right. \nonumber \\
 &\underbrace {\begin{array}{*{20}{c}}
   {} & {} & {} & {} & {} & {} & {} & {} & {} &{}  & {\left. { - {F_{{X_i}}}\left( {{\tau }\left( {\rho \varpi z + 1} \right)} \right)} \right]dydz}  \\
\end{array}}_{{J_{23}}} ,
\end{align}

On the basis of Appendix A, for an arbitrary choice of $\alpha$, we can use Gaussian-Chebyshev quadrature to find the approximation for the CDF of $X_{i}$ in \eqref{approximation expression for CDF}.
In addition, ${d_{{R_iD_1}}} = \sqrt {d_1^2+d_{S{R_i}}^2 - 2{d_1}{d_{S{R_i}}}\cos \left( {{\theta _i}} \right)} $ and
 ${d_{{R_iD_1}}} \gg d_{SR_{i}}$, we can approximate the distance between $R_i$ and ${D_1}$ as ${d_{{R_iD_1}}} \approx d_1$. The approximation for pdf of $Y_{1i}$ is given by
\begin{align}\label{approximation expression for Y1i}
{f_{{Y_{1i}}}}\left( y \right) \approx 1 - {e^{ - \left( {1 + d_1^\alpha } \right){\tau }}}.
\end{align}

Substituting \eqref{approximation expression for CDF} and \eqref{approximation expression for Y1i} into \eqref{the expression of Xi}, $J_{22}$ and $J_{23}$ can be calculated as follows:
\begin{align}\label{approx J22}
 {J_{22}} \approx &  {e^{ - \left( {1 + d_1^\alpha } \right)\theta }}\frac{\pi }{{2N}}\sum\limits_{n = 1}^N {\sqrt {1 - \phi _n^2} } \left( {{\phi _n} + 1} \right) \nonumber \\
 & \times \int_0^\infty  {\frac{1}{{{\Omega _{{\rm{LI}}}}}}{e^{ - \frac{z}{{{\Omega _{{\rm{LI}}}}}}}}\left( {{e^{ - {c_n}{\tau }\left( {\rho \varpi  z + 1} \right)}} - {e^{ - {c_n}\xi \left( {\rho \varpi z + 1} \right)}}} \right)} dz \nonumber \\
  %=& {e^{ - \left( {1 + d_1^\alpha } \right)\theta }}\frac{\pi }{{2N}}\sum\limits_{n = 1}^N {\sqrt {1 - \phi _n^2} } \left( {{\phi _n} + 1} \right) \nonumber \\
  %&\times \left( {\frac{{{e^{ - {c_n}{\tau }}}}}{{{1 + \rho \varpi \tau {c_n}{\Omega _{{\rm{LI}}}}}}} - \frac{{{e^{ - {c_n}\xi }}}}{{{1 + \rho \varpi \xi {c_n}{\Omega _{{\rm{LI}}}}}}}} \right)\nonumber \\
   =& {e^{ - \left( {1 + d_1^\alpha } \right)\theta }}\frac{\pi }{{2N}}\sum\limits_{n = 1}^N {\sqrt {1 - \phi _n^2} } \left( {{\phi _n} + 1} \right)\nonumber \\
   &\times \left[ {\chi {e^{ - {c_n}\tau }} - \psi {e^{ - {c_n}\xi }}} \right],
\end{align}
where $\chi  = \frac{1}{{{1 + \rho \varpi \tau {c_n}{\Omega _{{\rm{LI}}}}}}}$ and $\psi  = \frac{1}{{{1 + \rho \varpi \xi {c_n}{\Omega _{{\rm{LI}}}}}}}$.
\begin{align}\label{approx J23}
{J_{23}} \approx &   \frac{{\pi \left( {1 + d_1^\alpha } \right)}}{{2N}}\sum\limits_{n = 1}^N {\sqrt {1 - \phi _n^2} } \left( {{\phi _n} + 1} \right) \nonumber\\
  & \times \int_0^\infty  {\frac{1}{{{\Omega _{{\rm{LI}}}}}}{e^{ - \frac{z}{{{\Omega _{{\rm{LI}}}}}}}}} \int_{{\tau }}^{\xi}  {\left( {{e^{ - {c_n}{\tau }\left( {\rho \varpi z + 1} \right) - \left( {1 + d_1^\alpha } \right)y}}} \right.} \nonumber \\
 &\left. { - {e^{ - {c_n}y\left( {\rho \varpi z + 1} \right) - \left( {1 + d_1^\alpha } \right)y}}} \right)dydz \nonumber\\
 =  & \Delta \chi {e^{ - {c_n}{\tau }}}\left( {{e^{ - \left( {1 + d_1^\alpha } \right){\tau }}} - {e^{ - \left( {1 + d_1^\alpha } \right)\xi }}} \right) \nonumber\\
  & + \Delta {\rm T}\underbrace {\int_0^\infty  {\frac{1}{{z + \zeta }}} {e^{ - \frac{z}{{\psi {\Omega _{{\rm{LI}}}}}}}}dz}_{{I_1}} - \Delta \Phi \underbrace {\int_0^\infty  {\frac{1}{{z + \zeta }}{e^{ - \frac{z}{{\chi {\Omega _{{\rm{LI}}}}}}}}dz} }_{{I_2}},
\end{align}
where $\Delta {\rm{ = }}\frac{\pi }{{2N}}\sum\limits_{n = 1}^N {\sqrt {1 - \phi _n^2} } \left( {{\phi _n} + 1} \right)$, $\zeta  = \frac{{{c_n} + \left( {1 + d_1^\alpha } \right)}}{{{\rho \varpi {c_n}}}}$, ${\rm T} = \frac{{\left( {1 + d_1^\alpha } \right){e^{ - \left( {{c_n} + \left( {1 + d_1^\alpha } \right)} \right)\xi }}}}{{{\rho \varpi {c_n}{\Omega _{{\rm{LI}}}}}}}$ and $\Phi  = \frac{{\left( {1 + d_1^\alpha } \right){e^{ - \left( {{c_n} + \left( {1 + d_1^\alpha } \right)} \right){\tau }}}}}{{{\rho \varpi {c_n}{\Omega _{{\rm{LI}}}}}}}$.

By the virtue of  \cite[Eq. (3.352.4)]{gradshteyn}, $I_{1}$ and $I_{2}$ can be given by
\begin{align}\label{I1}
{I_1} =  - {e^{\frac{\zeta }{{\psi {\Omega _{{\rm{LI}}}}}}}}{{\mathop{\rm E}\nolimits} _{\mathop{\rm i}\nolimits} }\left( { - \frac{\zeta }{{\psi {\Omega _{{\rm{LI}}}}}}} \right),
\end{align}
and
\begin{align}\label{I2}
{I_2} =  - {e^{\frac{\zeta }{{\chi {\Omega _{{\rm{LI}}}}}}}}{{\mathop{\rm E}\nolimits} _{\mathop{\rm i}\nolimits} }\left( { - \frac{\zeta }{{\chi {\Omega _{{\rm{LI}}}}}}} \right),
\end{align}
respectively.

Substituting \eqref{I1} and \eqref{I2} into \eqref{approx J23}, we can obtain
\begin{align}\label{the last expression for J23}
 {J_{23}} \approx & \Delta \chi {e^{ - {c_n}{\tau }}}\left( {{e^{ - \left( {1 + d_1^\alpha } \right){\tau }}} - {e^{ - \left( {1 + d_1^\alpha } \right)\xi }}} \right)\nonumber \\
 & - \Delta {\rm T}{e^{\frac{\zeta }{{\psi {\Omega _{{\rm{LI}}}}}}}}{{\mathop{\rm E}\nolimits} _{\mathop{\rm i}\nolimits} }\left( {\frac{{ - \zeta }}{{\psi {\Omega _{{\rm{LI}}}}}}} \right) + \Delta \Phi {e^{\frac{\zeta }{{\chi {\Omega _{{\rm{LI}}}}}}}}{{\mathop{\rm E}\nolimits} _{\mathop{\rm i}\nolimits} }\left( {\frac{{ - \zeta }}{{\chi {\Omega _{{\rm{LI}}}}}}} \right).
\end{align}

Applying the results derived in Appendix A, the denominator $\Xi_{2} $ for $J_{21}$ in \eqref{derived process for J1 in lemma} can be approximated as follows:
\begin{align}\label{the expression of Xi2}
{\Xi _2} \approx {e^{ - \left( {1 + d_1^\alpha } \right){\tau }}}\left( {1 - \Delta \left( {1 - \chi {e^{ - {c_n}{\tau }}}} \right)} \right).
\end{align}

Combining \eqref{approx J22}, \eqref{the last expression for J23} and \eqref{the expression of Xi2}, we can obtain
\begin{align}\label{the expression of J21}
 {J_{21}} \approx &\frac{{{e^{ - \left( {1 + d_1^\alpha } \right)\theta }}\Delta \left( {\chi {e^{ - {c_n}{\tau }}} - \psi {e^{ - {c_n}\xi }}} \right)}}{{{e^{ - \left( {1 + d_1^\alpha } \right){\tau }}}\left( {1 - \Delta \left( {1 - \chi {e^{ - {c_n}{\tau }}}} \right)} \right)}} \nonumber \\
  &+ \frac{{\Delta \chi {e^{ - {c_n}{\tau }}}\left( {{e^{ - \left( {1 + d_1^\alpha } \right){\tau }}} - {e^{ - \left( {1 + d_1^\alpha } \right)\xi }}} \right)}}{{{e^{ - \left( {1 + d_1^\alpha } \right){\tau }}}\left( {1 - \Delta \left( {1 - \chi {e^{ - {c_n}{\tau }}}} \right)} \right)}}\nonumber \\
  &- \frac{{\Delta {\rm T}{e^{\frac{\zeta }{{\psi {\Omega _{{\rm{LI}}}}}}}}}}{{{e^{ - \left( {1 + d_1^\alpha } \right){\tau}}}\left( {1 - \Delta \left( {1 - \chi {e^{ - {c_n}{\tau }}}} \right)} \right)}}{{\mathop{\rm E}\nolimits} _{\mathop{\rm i}\nolimits} }\left( {\frac{{ - \zeta }}{{\psi {\Omega _{{\rm{LI}}}}}}} \right)\nonumber \\
  &+ \frac{{\Delta \Phi {e^{\frac{\zeta }{{\chi {\Omega _{{\rm{LI}}}}}}}}}}{{{e^{ - \left( {1 + d_1^\alpha } \right){\tau}}}\left( {1 - \Delta \left( {1 - \chi {e^{ - {c_n}{\tau }}}} \right)} \right)}}{{\mathop{\rm E}\nolimits} _{\mathop{\rm i}\nolimits} }\left( {\frac{{ - \zeta }}{{\chi {\Omega _{{\rm{LI}}}}}}} \right).
\end{align}

Similarly as the above derived process, we can obtain
\begin{align}\label{the expression of J22}
 {J_{31}} \approx & \frac{{{e^{ - \left( {1 + d_1^\alpha } \right){\tau }}} - {e^{ - \left( {1 + d_1^\alpha } \right)\xi }}}}{{{e^{ - \left( {1 + d_1^\alpha } \right){\tau }}}\left( {1 - \Delta \left( {1 - \chi {e^{ - {c_n}{\tau }}}} \right)} \right)}} \nonumber\\
  -& \frac{{\Delta \left( {{e^{ - \left( {1 + d_1^\alpha } \right){\tau }}} - {e^{ - \left( {1 + d_1^\alpha } \right)\xi }}} \right)}}{{{e^{ - \left( {1 + d_1^\alpha } \right){\tau }}}\left( {1 - \Delta \left( {1 - \chi {e^{ - {c_n}{\tau }}}} \right)} \right)}} \nonumber\\
  +& \frac{{\Delta \left( {{\rm T} {e^{\frac{\zeta }{{{\Omega _{{\rm{LI}}}}\psi }}}}{{\mathop{\rm E}\nolimits} _{\mathop{\rm i}\nolimits} }\left( {\frac{{ - \zeta }}{{{\Omega _{{\rm{LI}}}}\psi }}} \right) - \Phi {e^{\frac{{ - \zeta }}{{{\Omega _{{\rm{LI}}}}\chi }}}}{{\mathop{\rm E}\nolimits} _{\mathop{\rm i}\nolimits} }\left( {\frac{{ - \zeta }}{{{\Omega _{{\rm{LI}}}}\chi }}} \right)} \right)}}{{{e^{ - \left( {1 + d_1^\alpha } \right){\tau }}}\left( {1 - \Delta \left( {1 - \chi {e^{ - {c_n}{\tau }}}} \right)} \right)}}.
\end{align}

Combining \eqref{the expression of J21} and \eqref{the expression of J22}, we can obtain \eqref{the approx express Theta1}.
The proof is completed.

\appendices
\section*{Appendix~C: Proof of Corollary \ref{asymptotic OP for FD TRS}} \label{Appendix:C}
\renewcommand{\theequation}{C.\arabic{equation}}
\setcounter{equation}{0}

Based on the derived results in Appendix B, the proof starts by providing the term $J_{22}$ with $\varpi=1$  as follows:
\begin{align}\label{the copy of approx J22}
 {J_{22}}\approx& {e^{ - \left( {1 + d_1^\alpha } \right)\theta }}\frac{\pi }{{2N}}\sum\limits_{n = 1}^N {\sqrt {1 - \phi _n^2} } \left( {{\phi _n} + 1} \right) \nonumber \\
  &\times \left( {\frac{{{e^{ - {c_n}{\tau}}}}}{{1 + \rho{\tau}{c_n} {\Omega _{{\rm{LI}}}}}} - \frac{{{e^{ - {c_n}\xi }}}}{{1 + \rho\xi {c_n}  {\Omega _{{\rm{LI}}}}}}} \right).
\end{align}

To facilitate our asymptotic analysis, when $x  \to 0$, we use zero order series expansion to approximate the exponential function ${e^x}$, i.e., ${e^x} \approx 1$. Therefore, $J_{22}$ can be further approximated as follows:
\begin{align}\label{the copy of approx J22 further appro}
 {J_{22}}\approx& \frac{\pi }{{2N}}\sum\limits_{n = 1}^N {\sqrt {1 - \phi _n^2} } \left( {{\phi _n} + 1} \right) \nonumber \\
  &\times \left( {\frac{{1}}{{1 + \rho{\tau}{c_n} {\Omega _{{\rm{LI}}}}}} - \frac{{1}}{{1 +\rho\xi  {c_n} {\Omega _{{\rm{LI}}}}}}} \right).
\end{align}
Similar as \eqref{the copy of approx J22 further appro}, $J_{23}$ and $J_{31}$ can be further approximated by utilizing zero order series expansion as follows:
\begin{align}\label{the copy of approx J23 further appro}
 J_{23} \approx 0,
\end{align}
and
\begin{align}\label{the copy of approx J31 further appro}
 J_{31} \approx 0,
\end{align}
respectively.
Additionally, the denominator $\Xi_{2} $ for $J_{21}$ in \eqref{derived process for J1 in lemma},
\begin{align}\label{the copy of Xi2 further appro}
{\Xi _2} \approx 1 - \Delta \left( {1 - \chi } \right),
\end{align}

Substituting \eqref{the copy of approx J22 further appro}, \eqref{the copy of Xi2 further appro}, \eqref{the copy of approx J23 further appro} and \eqref{the copy of approx J31 further appro} into \eqref{derived process for Theta1 in lemma}, the conditional probability ${\Theta _1}$ can be obtained as follows:
\begin{align}\label{the copy of approx Theta1 further appro}
{\Theta _1} \approx \frac{{\Delta \left( {\chi  - \psi } \right)}}{{1 - \Delta \left( {1 - \chi } \right)}}.
\end{align}

Using a similar approximation method as that used to obtain \eqref{the copy of approx Theta1 further appro}, ${\Theta _2}$ is given by
\begin{align}\label{the copy of approx Theta2 further appro}
{\Theta _2} \approx {K \choose
  k  }{\left( {1 - \left( {1 - \Delta \left( {1 - \chi } \right)} \right)} \right)^{K - k}}{\left( {1 - \Delta \left( {1 - \chi } \right)} \right)^k}.
\end{align}

Substituting \eqref{the copy of approx Theta1 further appro}, \eqref{the copy of approx Theta2 further appro} into \eqref{TS OP expression} and applying some manipulations, we can obtain \eqref{asymptotic OP for TS relay selection FD mode}.
The proof is completed.

\bibliographystyle{IEEEtran}
\bibliography{mybib}

\end{document}